\documentclass[english,twocolumn,transaction,10pt]{IEEEtran}
\usepackage[T1]{fontenc}
\usepackage{float}
\usepackage{amsmath}
\usepackage{amsthm}
\usepackage{amssymb}
\usepackage{stackrel}
\usepackage{graphicx}
\usepackage{setspace}
\usepackage{url}
\usepackage{bm}
\usepackage{caption}
\usepackage{subfigure}
\usepackage{cases}
\usepackage{xcolor}
\usepackage{diagbox}
\usepackage{algorithm}
\usepackage{algorithmic}
\usepackage{setspace}
\usepackage{stfloats}
\allowdisplaybreaks[4]
\renewcommand{\algorithmicrequire}{ \textbf{Input:}} 
\renewcommand{\algorithmicensure}{ \textbf{Output:}} 
\usepackage{multirow}
\usepackage{makecell}
\usepackage{threeparttable}
\usepackage{setspace}

\makeatletter

\floatstyle{ruled}
\newfloat{algorithm}{tbp}{loa}
\providecommand{\algorithmname}{Algorithm}
\floatname{algorithm}{\protect\algorithmname}

\theoremstyle{plain}

\theoremstyle{definition}

\theoremstyle{plain}

\theoremstyle{plain}

\IEEEoverridecommandlockouts
\usepackage{ifpdf}
\usepackage{cite}
\ifCLASSOPTIONcompsoc
\usepackage[caption=false,font=normalsize,labelfont=sf,textfont=sf]{subfig}
\else
\usepackage[caption=false,font=footnotesize]{subfig}
\fi
\usepackage{amsmath}
\usepackage{mathtools}
\usepackage{booktabs}
\usepackage{multirow}
\usepackage{setspace}
\usepackage{lettrine}
 \usepackage{array}
\@ifundefined{showcaptionsetup}{}{
	\PassOptionsToPackage{caption=false}{subfig}}
\usepackage{subfig}
\makeatother
\usepackage{babel}

\newtheorem{proposition}{Proposition}

\usepackage{subfigure}
\usepackage[justification=centering]{caption}
\usepackage{caption}
\newcommand{\setParDef}{\setlength {\parskip} {0pt} }

\def\BibTeX{{\rm B\kern-.05em{\sc i\kern-.025em b}\kern-.08em
    T\kern-.1667em\lower.7ex\hbox{E}\kern-.125emX}}

\begin{document}

%

\title{
A Reconfigurable Subarray Architecture and Hybrid Beamforming for 
Millimeter-Wave Dual-Function-Radar-Communication Systems
}

\author{
 Xin~Jin, Tiejun~Lv,~\IEEEmembership{Senior Member,~IEEE}, Wei Ni,~\IEEEmembership{Fellow,~IEEE}, Zhipeng~Lin,~\IEEEmembership{Member,~IEEE}, Qiuming~Zhu,~\IEEEmembership{Senior Member,~IEEE}, Ekram~Hossain,~\IEEEmembership{Fellow,~IEEE},~and~H.~Vincent~Poor,~\IEEEmembership{Life~Fellow,~IEEE}

\thanks{Manuscript received March 6, 2024; accepted April 22, 2024. This paper was supported in part by the National Natural Science Foundation of China under No. 62271068, the Beijing Natural Science Foundation under Grant No. L222046, the U.S National Science Foundation under Grants CNS-2128448 and ECCS-2335876. (\emph{corresponding author: Tiejun Lv}.)}

\thanks{X. Jin and T. Lv are with the School of Communication and Information Engineering, Beijing University of Posts and Telecommunications, Beijing 100876, China (e-mail: \{jxzoe, lvtiejun\}@bupt.edu.cn).

W. Ni is with the Data61, Commonwealth Scientific and Industrial Research Organization, Sydney, Marsfield, NSW 2122, Australia.
(e-mail: wei.ni@ieee.org).

Z. Lin and Q. Zhu are with the College of Electronic and Information Engineering, Nanjing University of Aeronautics and Astronautics, Nanjing 211106, China (e-mail: \{linlzp, zhuqiuming\}@nuaa.edu.cn).

E. Hossain is with the Department of Electrical and Computer Engineering, University of Manitoba, Canada (e-mail: ekram.hossain@umanitoba.ca).

H. V. Poor is with the Department of Electrical and Computer Engineering, Princeton University, Princeton, NJ 08544 USA (e-mail: poor@princeton.edu).   
}
}

\maketitle

\begin{abstract}
Dual-function-radar-communication (DFRC) is a promising candidate technology for next-generation networks. By integrating hybrid analog-digital (HAD) beamforming into a multi-user millimeter-wave (mmWave) DFRC system, we design a new reconfigurable subarray (RS) architecture and jointly optimize the HAD beamforming to maximize the communication sum-rate and ensure a prescribed signal-to-clutter-plus-noise ratio for radar sensing. Considering the non-convexity of this problem arising from multiplicative coupling of the analog and digital beamforming, we convert the sum-rate maximization into an equivalent weighted mean-square error minimization and apply penalty dual decomposition to decouple the analog and digital beamforming. Specifically, a second-order cone program is first constructed to optimize the fully digital counterpart of the HAD beamforming. Then, the sparsity of the RS architecture is exploited to obtain a low-complexity solution for the HAD beamforming. The convergence and complexity analyses of our algorithm are carried out under the RS architecture.  Simulations corroborate that, with the RS architecture, DFRC offers effective communication and sensing and improves energy efficiency by 83.4\% and 114.2\% with a moderate number of radio frequency chains and phase shifters, compared to the persistently- and fully-connected architectures, respectively.
\end{abstract}


\IEEEpeerreviewmaketitle

\section{Introduction}

\IEEEPARstart{T}{he} sixth generation (6G) of wireless systems is anticipated to go beyond traditional mobile networks in continuous and ubiquitous sensing capabilities~\cite{9737357}. A promising concept for realizing simultaneous sensing and communication capabilities is dual-functional-radar-communication (DFRC). By sharing hardware, software, and information, DFRC enables these capabilities while coping with resource constraints~\cite{10188491,9303435}. Benefiting from advances in millimeter-wave (mmWave) technologies, mmWave DFRC can support simultaneous high-speed communications and high-resolution sensing in environment-aware scenarios~\cite{9851407,9606831}, e.g., smart homes and Internet of Things (IoT). For instance, the authors of \cite{9851407} proposed a novel dual-function waveform design scheme that can estimate velocity and range at ultra-high resolutions in noisy environments.

To address severe path-loss of mmWave signals, massive multiple-input multiple-output (mMIMO) antenna arrays with significant beamforming gains have been studied for both communication and radar systems~\cite{8241348,8962251,zhang2023doppler}. 
The authors of \cite{zhang2023doppler} provided a parameter design approach for high mobility scenarios to enhance the resilience to Doppler shift.
Unfortunately, a fully digital (FD) mMIMO architecture is of limited practical value due to its prohibitive hardware and energy costs. To this end, 
hybrid analog-digital (HAD) architectures, offering flexible connectivity between radio frequency (RF) chains and antennas, become promising to balance hardware complexity and system performance \cite{8371237,7397861,7389996}. 

A similar concept in radar systems, i.e., a phased-MIMO radar, is used to join the high resolution of a MIMO radar and the coherent processing gain of a phased-array radar~\cite{5419124,9237135}. For instance, the authors of~\cite{5419124} developed a hybrid beamforming (HBF) phased-array architecture for active sensing applications, using a hybrid beamformer system for transmitting and receiving images. In~\cite{9237135}, an approach to synthesize probing beampatterns for mMIMO radars with fully-connected (FC) structures was presented based on learning-based RF chains and antenna selection.

As communication and sensing converge in mmWave mMIMO systems, it is imperative to design an effective transceiver architecture to avoid high hardware costs and enhance the integration gain of DFRC systems. To this end, we present a new DFRC system with a reconfigurable subarray (RS) architecture for improved spectral efficiency (SE), energy efficiency (EE), and sensing performance.

\subsection{Related Works}

Many works have considered HAD DFRC systems, including FC  \cite{9729809,10199478,10049999,9538857,10279443} and persistently-connected (PC) architectures \cite{8683591,9500661,9366836,9950549}. In \cite{9729809}, two-stage alternating minimization was proposed to optimize hybrid transmit beams and phase vectors for detecting several targets and communicating with several users, subject to the quality-of-service (QoS) requirements, e.g.,  signal-to-interference-plus-noise ratio (SINR), for communication. The authors of \cite{10049999} conducted weighted optimization of optimal FD transmission and radar beampattern for an Internet of Vehicles (IoV) scenario to improve SE of vehicular communication as well as sensing accuracy. However, these studies focused on radar similarity constraints or beampattern approximation, and overlooked the sensing performance of the radar receiver. Considering that the radar detection probability directly depends on the signal-to-clutter-plus-noise ratio (SCNR) of the radar receiver, the authors of \cite{9538857} optimized hybrid transmit beamforming to maximize the receive SCNR. Additionally, the authors of \cite{10279443} opted for a communication-centric transmitter design under the Cramer-Rao bound (CRB). These studies were generally based on an FC architecture, where every RF chain is hardwired to all antennas via a phase shift (PS) network, which could be prohibitive in mMIMO. 

Another hybrid architecture is the PC architecture, where RF chains connect individual antenna subarrays. In \cite{8683591}, an effective trade-off between sub-array MIMO radar and communication was achieved by minimizing a weighted sum of beamforming errors for both. The EE of the integrated system was maximized by activating a small number of RF chains \cite{9500661}. In \cite{9366836}, the communication rate was maximized under a sensing beampattern mismatch requirement in both single- and multi-user multi-carrier systems. More recently, the authors of \cite{9950549} considered a more realistic clutter interference environment and used a double-phase-shifter (DPS) structure, providing a well-tolerated compromise between the number of PSs and the performance gain of the DPS HBF. Although the PC architecture significantly reduces hardware costs compared to the FC architecture, these HBF designs offer no flexibility.

RS architecture-based HBF offers a promising solution to balance system performance and hardware costs in communication systems~\cite{9110865,7880698,8502711}. This is achieved by flexibly connecting an RF chain to a non-overlapping subarray through a switch network and PSs~\cite{9110865}. Hybrid subarrays were reconfigurable based on channel statistics in~\cite{7880698}, where a greedy algorithm was implemented to approach the SE achieved by exhaustive search. In multi-user scenarios, the authors of \cite{8502711} designed HBF based on sub-channel differences. However, these discussions were limited to communication scenarios. A summary of these existing studies is provided in Table I.

\begin{table}[t]
\scriptsize
\centering
\caption{Comparison between the proposed scheme and the existing HBF designs}
\label{tab:Comp}
\renewcommand\arraystretch{1.2}
    \resizebox{\linewidth}{!}{
\begin{tabular}{|c|c|c|c|c|c|}
\hline
\multirow{2}{*}{Ref.} & \multicolumn{3}{c|}{Architecture} & \multirow{2}{*}{Communication Metric} & \multirow{2}{*}{Sensing Metric} \\
\cline{2-4}
~ & FC & PC & RS &  ~ & ~  \\
\hline
\cite{9729809} & \checkmark & ~ & ~ & SINR  & Desired
beampattern \\
\hline
\cite{10199478} & \checkmark & ~ & ~ & Sum-rate  & Beampattern matching error  \\
\hline
\cite{10049999} & \checkmark & ~ & ~ & \makecell[c]{FD matching error; \\ Multiuser interference (MUI)}  & \makecell[c]{Desired beampattern; \\Power distribution error}  \\
\hline
\cite{9538857} & \checkmark & ~ & ~ & SCNR & SINR  \\
\hline
\cite{10279443} & \checkmark & ~ & ~ & Weighted sum-rate & CRB   \\
\hline
\cite{8683591} & ~ & \checkmark & ~ & FD matching error & Desired
beampattern  \\
\hline
\cite{9366836} & ~ & \checkmark & ~ & Weighted sum-rate & Waveform
similarity \\
\hline
\cite{9950549} & ~ & \checkmark & ~ & Sum-rate & SCNR\\
\hline
\cite{9110865} & ~ & ~ & \checkmark & Sum-rate & Not applicable (n.a.)  \\
\hline
\cite{7880698} & ~ & ~ & \checkmark & SE & n.a.  \\
\hline
\cite{8502711} & ~ & ~ & \checkmark & Sum-rate & n.a. \\
\hline
Ours & ~ & ~ & \checkmark & Sum-rate & SCNR  \\
\hline
\end{tabular}}
\end{table}

\subsection{Challenges and Contributions}
  
The incorporation of an RS architecture and HBF in a mmWave DFRC system is not straightforward. On the one hand, the tight multiplicative coupling of the analog and digital beamforming makes their joint optimization challenging. 
The need to selectively connect the RF chains and antennas in the RS architecture could result in an NP-hard mixed integer programming problem~\cite{9110865}. Also, a direct use of techniques, such as alternating optimization and block decomposition, is likely to get trapped in inferior solutions or even fail to converge. On the other hand, many existing DFRC designs focus on the transmit HBF, which can penalize the sensing performance in cluttered environments. To cope with clutter interference, the authors of~\cite{9537599} considered the processing of radar receive filters and their optimization on a DFRC system, which effectively utilizes the available degrees of freedom (DoFs) in radar receive arrays. But only small-scale FD beamforming was considered in~\cite{9537599}.

This paper presents a new framework for HAD DFRC systems with hybrid transmit beamforming and radar receive beamforming, where the PSs and antennas are reconfigurably connected through adaptive switching networks. A dual-function base station (BS) delivers communication services to multiple users in the downlink, meanwhile detecting a potential target amidst signal-dependent interference sources,~i.e., clutter. Unlike conventional approaches focusing on radar beampattern matching errors, we adopt the radar receive SCNR to quantify and constrain the sensing performance. We propose an effective algorithm that optimizes the digital and analog beamformers to maximize the achievable sum-rate of the HAD DFRC system under a prescribed SCNR requirement of its radar sensing performance.

Following is the list of contributions of this paper:
\begin{itemize}
	\item We propose a new RS architecture for mmWave HAD DFRC systems in which every RF chain can be flexibly linked to a non-overlapping subset of antennas. This architecture can significantly reduce the hardware costs while maintaining the DFRC performance. 

    \item We present a new problem formulation to maximize the non-convex sum-rate under the SCNR constraints of radar sensing in an HAD DFRC system with an RS architecture. 
 
	\item Using penalty dual decomposition (PDD) and weighted minimum mean square error (WMMSE) as an interpretation of the non-convex sum-rate, we put forth a new WMMSE-PDD (WPDD) algorithm to solve the new problem through block coordination descent (BCD). As per iteration, a second-order cone program (SOCP) optimizes the FD transmit counterpart of the HBF. The sparsity of the RS architecture is exploited to obtain a low-complexity solution for the HBF.

    \item The proposed WPDD algorithm is extended to DFRC systems with PC architectures. Its convergence and complexity are analyzed, indicating the benefit of HBF in DFRC systems in terms of EE. 
\end{itemize}

Extensive simulations are carried out to verify our analysis and assess the effectiveness of a DFRC system with an RS architecture between performance and cost by comparing it with the FC and PC architectures. RS-based HBF leads to a 83.4\% improvement in EE compared to the PC-based HBF, and surpasses the EE achieved under the FC architectures by more than twofold.\par

The subsequent sections of this paper are arranged as follows. 
The studied problem is cast in Section II. The new WPDD algorithm is presented in Section III, followed by its convergence, complexity, and EE analyses in Section IV. A summary of simulation results is given in Section~V. Lastly, the conclusions are presented in Section~VI. 

\textit{Notation}: 
Bold-face upper- and lower-cases indicate matrices and vectors, respectively. $(\cdot)^T$, 
$(\cdot)^H$, $(\cdot)^\dag$, and $\text{tr}( \cdot )$ denote transpose, 
conjugate transpose, pseudo-inverse, and trace, respectively, ${\mathbb{C}^{M \times N}}$ denotes all $M \times N$ complex matrices,
${\cal R}{ \{  \cdot \}} $ takes the real part of a complex value, $\| \cdot\|$ provides Euclidean or Frobenius norm for a vector or matrix, $\| \cdot\|_0$ is the 0-norm of a vector, $\mathbb{E(\cdot)} $ takes the statistical expectation, ${\cal {CN}} (\cdot,\cdot )$ represents the complex normal distribution, blkdiag($\cdot$) denotes block diagonal matrix, $\left\lceil \cdot \right\rceil$ stands for ceiling. The notation used is collated in Table~\ref{tab:symbol}.  \par

\begin{table}[t] 
    \centering
    \caption{Notation and definitions}  \label{tab:symbol}
    \renewcommand\arraystretch{1.2}
    \resizebox{\linewidth}{!}{
    \begin{tabular}{c|l}
        \hline
        \textbf{Notation} & \textbf{Definition}\\   
        \hline
         ${M_{\mathrm{T}}}/{M_{\mathrm{R}}}$ & Number of antennas at the DFRC transmitter/receiver; \\
         $N_{\mathrm{RF}}^{\mathrm{t}}/N_{\mathrm{RF}}^{\mathrm{r}}$ & 
         Number of RF chains at the DFRC transmitter/receiver;  \\
         $K$ & Number of users;  \\
         $d_s$ & Number of data streams per user; \\
         ${M_{\mathrm{U}}}$ & Number of receive antennas at the $k$-th user;  \\
         ${P_{\mathrm{T}}}$ & Total transmit power of the DFRC BS;  \\
         $\alpha_{k,l}$ & The complex amplitude of the $l$-th path for the $k$-th user;\\
         ${\beta _0}/{\beta _j}$ & The complex reflection coefficients of the target/clutter; \\
         ${{\bf{T}}_{\mathrm{D}}}/{{\bf{W}}_{\mathrm{D}}}$ & Digital transmit/receive beamformer at the DFRC BS; \\
         ${{\bf{T}}_{\mathrm{RF}}}/{{\bf{W}}_{\mathrm{RF}}}$ & Analog transmit/receive beamformer at the DFRC BS;  \\
         ${\bf{U}}_k$  & 
         Receive beamformer at the $k$-th user;  \\   
         ${{\bf{H}}_k}$ & 
         CSI from the BS to the $k$-th user; \\
         ${{\bf{a}}_{\mathrm{t}}}\left( {{\theta}} \right) / {{\bf{a}}_{\mathrm{r}}} \left( {\theta} \right)$ & Transmit/receive steering vector at the angle $\theta$; \\
         $R$ & The achievable sum-rate of $K$ users; \\
         $\varGamma$ & The receive SCNR at the DFRC receiver; \\
         ${{\bf{T}}_k}/{\bf{W}}$ & Auxiliary variables introduced for the WPDD algorithm; \\
         $\rho$ & The penalty parameter of the WPDD algorithm; \\
         ${\bf{D}}_k/\tilde{\bf{D}}$ & The dual variables of the WPDD algorithm. \\
         \hline
    \end{tabular}   }
\end{table}

\section{System Architecture And Problem Statement}

\begin{figure}[t]
	\centering
	\includegraphics[scale=0.5]{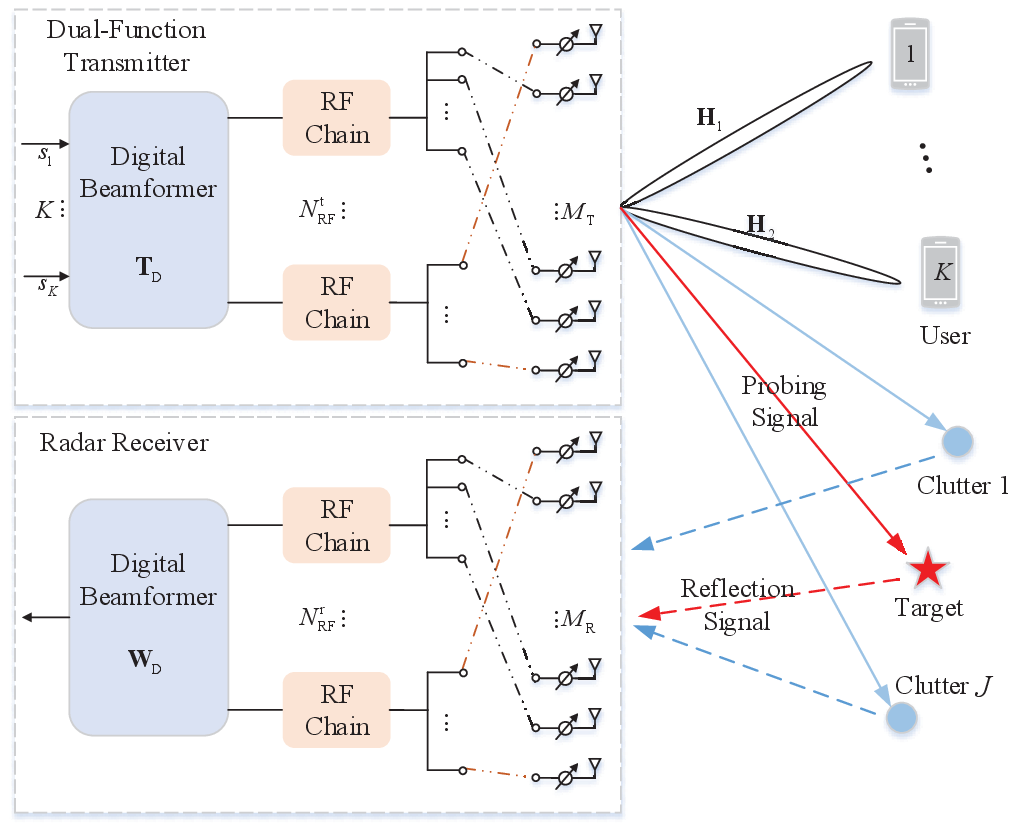}    \captionsetup{justification=raggedright,font={small}}
	\caption{Illustration of the considered DFRC system with an RS architecture.} 
	\label{system1_fig}
\end{figure}

As depicted in Fig. \ref{system1_fig}, an mmWave DFRC BS and users both have uniform linear arrays (ULAs) installed.
The DFRC BS has $M_{\mathrm{T}}$ transmit antennas and $N_{\mathrm{RF}}^{\mathrm{t}}$ RF chains to serve $K$ users in the downlink. It also has $M_{\mathrm{R}}$ receive antennas and $N_{\mathrm{RF}}^{\mathrm{r}}$ RF chains to detect a target of interest in cluttered environments. Each user has $M_{\mathrm{U}}$ receive antennas. 
The angle-of-departure (AoD) of the radar target is the same as the angle-of-arrival (AoA), as considered in the literature~\cite{6649991}. An RS-based HAD architecture is considered at both the receive and transmit sides of the DFRC BS, where every RF chain can be linked to a non-overlapping subset of transmit antennas via an adaptive switching network. 

Let ${\bf{s}}( t ) = {\left[ {{\bf{s}}_1^T( t ),\cdots,{\bf{s}}_K^T( t )} \right]^T} \in {\mathbb{C}^{N_s}}$ ($N_s \le N_{\mathrm{RF}}^{\mathrm{t}} \ll {M_{\mathrm{T}}}$) denote the symbol sequence sent in the $t$-th slot for $K$ users and the radar target. $d_s$ is the number of data streams for user $k$ ($\mathbb{E} \{ {{\bf{s}}_k( t ){\bf{s}}_k{{( t )}^H}} \} = {{\bf{I}}_{d_s}}$). The symbols are independent among the users. The transmit signal is obtained as
\begin{equation} \label{xsignal}
{\bf{x}} (t)  = {{\bf{T}}_{\mathrm{RF}}}{{\bf{T}}_{\mathrm{D}}}{\bf{s}} ( t ) = \sum\limits_{k = 1}^K {{{\bf{T}}_{\mathrm{RF}}}{{\bf{T}}_{{\mathrm{D}},k}}{{\bf{s}}_k}( t )},
\end{equation}
where ${\bf{T}}_{{\mathrm{D}},k} \in {\mathbb{C}^{{N_{\mathrm{RF}}^{\mathrm{t}}} \times d_s}}$ stands for the digital beamfomer for the $k$-th user, ${{\bf{T}}_{\mathrm{D}}} = \left[ {{{\bf{T}}_{{\mathrm{D}},1}},\cdots,{{\bf{T}}_{{\mathrm{D}},K}}} \right]$, and ${\mathbf{T}}_{\mathrm{RF}} \in {\mathbb{C}^ {M_{\mathrm{T}} \times N_{\mathrm{RF}}^{\mathrm{t}}}}$ stands for the analog beamformer.

\subsection{Communication Model}
Following the extended Saleh-Valenzuela model~\cite{8683591,6834753}, mmWave channel between the DFRC BS and user $k$ is modeled as  
\begin{equation} \label{hk}
{{\bf{H}}_k} = \sqrt {\frac{{{M_{\mathrm{T}}}{M_{\mathrm{U}}}}}{{{L_k}}}} \sum\limits_{l = 1}^{{L_k}} {{\alpha _{k,l}} {\bf{a}}_{\mathrm{r}} ({\psi _{k,l}} ) {\bf{a}}_{\mathrm{t}}^H \left( {{\phi _{k,l}}} \right)},\forall k,
\end{equation}
where $L_k$ specifies the number of propagation paths to the user $k$, $\alpha_{k,l} \sim {\cal {CN}} ( {0,10^{-0.1PL(d_k)}} ) $ is the complex gain of the $l$-th path, $l= 1,\cdots,L_k$. $PL(d_k)$ is the path loss over the distance, $d_k$, between the BS and user $k$. ${\bf{a}}_{\mathrm{t}} ( {{\phi _{k,l}}} )$ and ${\bf{a}}_{\mathrm{r}} ( {{\psi _{k,l}}} )$ give the transmit and receive array response vectors, respectively:
\begin{equation}     
 {\bf{a}}_{\mathrm{t}} (\! {{\phi _{k,l}}}  ) \!\! =  \! \!\frac{1}{{\sqrt { M_{\mathrm{T}} } }}{ {\! \left[ \!{1, {e^{ - \!j\frac{{2\pi }}{\lambda } \! d\sin\!  {{\phi _{k,l}}}  }} \!,\!  \cdots \!,\! {e^{ - \!j\frac{{2\pi }}{\lambda } \! d\left( \! {M_{\mathrm{T}} \! - \! 1} \!\right) \sin\! {{\phi _{k,l}}}  }}} \! \right] \!} ^T}, 
\end{equation}
\begin{equation} 
    { \bf{a}}_{\mathrm{r}} (\! {{\psi _{k,l}}}  ) \!\! =  \!\! \frac{1}{{\sqrt { M_{\mathrm{U}} } }}{ {\! \left[ \!{1, {e^{ - \!j\frac{{2\pi }}{\lambda } \! d\sin\!  {{\psi _{k,l}}}  }} \!,\!  \cdots \!,\! {e^{ -\! j\frac{{2\pi }}{\lambda } \!
    d\left( \! {M_{\mathrm{U}} \! - \! 1} \!\right) \!\sin\!  {{\psi _{k,l}}}  }}} \! \right] \!} ^T}\!. 
\end{equation}
Here, $\phi_{k,l}$ and $\psi_{k,l}$ are the AoD/AoA of the $l$-th path to the $k$-th user, respectively; $\lambda$ indicates the signal wavelength; $d = \lambda/2$ specifies antenna spacing. Suppose ${\bf{H}}_k$, $\forall k$, is known at the DFRC BS. In practice, the channels can be estimated, e.g., using the algorithms developed in \cite{8462337,9207745,8286860}.

The received signal of user $k$ is
\begin{equation}  \label{ysignal}
    {{\bf{y}}_k} ( t ) = {\bf{H}}_k \sum\limits_{k = 1}^K {{{\bf{T}}_{\mathrm{RF}}}{{\bf{T}}_{{\mathrm{D}},k}}{{\bf{s}}_k}( t )}  + {{\bf{n}}_k}( t ),
\end{equation}
where ${{\bf{n}}_k}( t ) \sim {\cal CN}( {{\bf{0}},{\sigma ^2}{\bf{I}}_{M_{\mathrm{U}}} } )$ stands for the additive white Gaussian noise (AWGN) of the $k$-th user. ${{\bf{n}}_k}( t )$ and ${{\bf{s}}_k}( t )$ are independent of each other. For the brevity of notation, we suppress the time index $t$ in the remainder of this paper. 

After FD receive beamforming, ${\bf{U}}_k \in {\mathbb{C}^ {M_{\mathrm{U}} \times d_s}} $, the detected signal of user $k$ is 
\begin{align}   \label{ysignal2}
    {\hat{\bf{s}}_k} 
    & \!=\! {\bf{U}}_k^H {{\bf{y}}_k} \nonumber \\
    & \!=\! \underbrace{ {\bf{U}}_k^H {\bf{H}}_k{{\bf{T}}_{\mathrm{RF}}}{{\bf{T}}_{{\mathrm{D}},k}}{{\bf{s}}_k} }_{\text{desired signal}}  \!+ \! \underbrace{\sum\limits_{i \ne k}^K \!{ {\bf{U}}_k^H {\bf{H}}_k{{\bf{T}}_{\mathrm{RF}}}{{\bf{T}}_{{\mathrm{D}},i}}{{\bf{s}}_i}} }_{\text{multi-user interference}}\!+\! \underbrace{{\bf{U}}_k^H {{\bf{n}}_k} }_{\text{noise}}.
\end{align}
Consequently, the overall achievable rate of the $K$ users is
\begin{align}
R (  {{{\bf{T}}_{\mathrm{RF}}} , \!{{\bf{T}}_{{\mathrm{D}},k}}}, {\bf{U}}_k )  \! = \!
&\sum\limits_{k = 1}^K {\log \det} ( {\bf{I}}_{d_s} + {\bf{U}}_k^H {\bf{H}}_k {{\bf{T}}_{\mathrm{RF}}}{{\bf{T}}_{{\mathrm{D}},k}} \nonumber \\
&\times {{\bf{T}}_{{\mathrm{D}},k}^H} {{\bf{T}}_{\mathrm{RF}}^H} { {\bf{H}}_k^H } {{\bf{U}}_k} {\bf{R}}_k^{-1}),
\end{align}
where ${\bf{R}}_k \!=\! {\bf{U}}_k^H ( {\sigma ^2}{\bf{I}}_{M_{\mathrm{U}}} \!+\! \sum\limits_{i \ne k}^K \!{ {\bf{H}}_k{{\bf{T}}_{\mathrm{RF}}}{{\bf{T}}_{{\mathrm{D}},i}} {{\bf{T}}_{{\mathrm{D}},i}^H} {{\bf{T}}_{\mathrm{RF}}^H} { {\bf{H}}_k^H }  }  ) {{\bf{U}}_k} $.

\subsection{Radar Model}
Apart from sending communication symbols to the $K$ users, the transmitted waveform is also exploited to accomplish a radar detection task. Considering an obstacle detection scenario, a target (e.g., a slowly moving pedestrian) is located at angle $\theta _0$. There are also $J$ stationary clutterers (e.g., trees and buildings) at angles $\theta _j, j = 1,\cdots,J$. In this case, the time delay and Doppler shift are relatively negligible, as discussed in \cite{9724259}. The prior knowledge of $\theta _0$ and $\theta _j$ is obtainable from an environmental dynamic database using a cognitive paradigm~\cite{6404093,9724259}. 

Suppose that the potential self-interference from the transmit array to the receive array of the MIMO radar is adequately addressed by implementing appropriate techniques, for instance, the one developed in \cite{6832464}. According to \eqref{xsignal}, the echo signal received by the BS can be written as

\begin{equation}  \label{rsignal1}
    {{\bf{y}}_{\mathrm{R}}} = {\beta _0}{{\bf{a}}_{\mathrm{r}}}( {{\theta _0}} ){\bf{a}}_{\mathrm{t}}^T( {{\theta _0}} ){\bf{x}}+ \sum\limits_{j = 1}^J {{\beta _j}{{\bf{a}}_{\mathrm{r}}}( {{\theta _j}} ){\bf{a}}_{\mathrm{t}}^T( {{\theta _j}} ){\bf{x}}} + {\bf{n}}_{\mathrm{R}},
\end{equation}
where $\beta _0$ and $\beta _j$ denote the round-trip attenuation and complex reflection coefficients concerning the target and the $j$-th clutter, respectively; 
$\mathbb{E}( {{\left| {\beta _0} \right|}^2} ) = \sigma _0^2$; $\mathbb{E}( {{\left| {{\beta _j}} \right|}^2} ) = \sigma _{\mathrm{C}}^2$; ${\bf{a}}_{\mathrm{r}} ( {\theta } )$ gives the $M_{\mathrm{R}} \times 1$ receive steering vector; ${\bf{a}}_{\mathrm{t}} ( {\theta} )$ specifies the $M_{\mathrm{T}} \times 1$ transmit steering vector; ${\bf{a}}_{\mathrm{r}} ( \theta ) $ and ${\bf{a}}_{\mathrm{t}}(\theta )$ are obtained in the same way as (3); ${\bf{n}}_{\mathrm{R}} \sim {\cal CN}( {{\bf{0}},\sigma _{\mathrm{R}}^2{{\bf{I}}_{M_{\mathrm{R}}}}} )$ is the AWGN. 

An RS-based HAD architecture is also considered at the DFRC BS receiver. The received echo signal after analog receive beamforming ${\mathbf{W}}_{\mathrm{RF}} \in {\mathbb{C}^ {M_{\mathrm{R}} \times N_{\mathrm{RF}}^{\mathrm{r}}}}$ and digital receive beamforming ${\bf{W}}_{{\mathrm{D}}} \in {\mathbb{C}^{N_{\mathrm{RF}}^{\mathrm{r}} \times N_{\mathrm{s}} }}$, is given by
\begin{align}  \label{rsignal2}
    {\bf{y}}'_{\mathrm{R}} 
    &= {{\bf{W}}_{{\mathrm{D}}}^H}{{\mathbf{W}}_{\mathrm{RF}}^H}{{\bf{y}}_{\mathrm{R}}} \nonumber \\
    &= \underbrace{{\beta _0}{{\bf{W}}_{{\mathrm{D}}}^H}{{\mathbf{W}}_{\mathrm{RF}}^H}{\bf{A}}( {{\theta _0}} ){\bf{x}} }_{\text{target}} + \underbrace{ {{\bf{W}}_{{\mathrm{D}}}^H}{{\mathbf{W}}_{\mathrm{RF}}^H} \sum\limits_{j = 1}^J \! {{\beta _j}{\bf{A}}( {{\theta _j}} ){\bf{x}}} }_{\text{clutters}} \nonumber \\
    &+ \underbrace{{{\bf{W}}_{{\mathrm{D}}}^H}{{\mathbf{W}}_{\mathrm{RF}}^H}{\bf{n}}_{\mathrm{R}} }_{\text{noise}},
\end{align}
where ${\bf{A}}(  \theta ) \! = \! {{\bf{a}}_{\mathrm{r}}}( \theta  ){\bf{a}}_{\mathrm{t}}^T( \theta )$.

We adopt the SCNR as the metric to measure the radar's target detection and localization capabilities, as given by
\begin{align}   \label{Rsnr}
   & \varGamma ( {{{\bf{T}}_{\mathrm{RF}}},{{\bf{T}}_{\mathrm{D}}},{\bf{W}}_{{\mathrm{RF}}},{\bf{W}}_{{\mathrm{D}}} } ) \nonumber \\
   & \!=\! \frac{ \mathbb{E} {\left( {\left\| {{\beta _0}{{\bf{W}}_{{\mathrm{D}}}^H}{{\mathbf{W}}_{\mathrm{RF}}^H}{\bf{A}}( {{\theta _0}} ){\bf{x}}} \right\|}^2 \right) } } { \mathbb{E} \left( {{{{\left\| {{{\bf{W}}_{{\mathrm{D}}}^H}{{\mathbf{W}}_{\mathrm{RF}}^H} \! \sum\limits_{j = 1}^J \!{{\beta _j}{\bf{A}}\left( {{\theta _j}} \right){\bf{x}}} } \right\|}^2}}} \! \right)+ \! \mathbb{E} \left( {\left\| {{\bf{W}}_{{\mathrm{D}}}^H}{{\mathbf{W}}_{\mathrm{RF}}^H}{\bf{n}}_{\mathrm{R}}\right\|}^2 \right)}  \nonumber \\ 
   &  \!=\! \frac{ \mathrm{tr} ({{{\bf{W}}_{{\mathrm{D}}}^H}{{\mathbf{W}}_{\mathrm{RF}}^H}{{\bf{\Sigma}}_{\mathrm{t}}} {{\bf{W}}_{{\mathrm{RF}}}}{{\mathbf{W}}_{\mathrm{D}}} }) }{ \mathrm{tr}({{{\bf{W}}_{{\mathrm{D}}}^H}{{\mathbf{W}}_{\mathrm{RF}}^H}{{\bf{\Sigma}}_{\mathrm{cn}}} {{\bf{W}}_{{\mathrm{RF}}}}{{\mathbf{W}}_{\mathrm{D}}} } ) }, 
\end{align}
where ${\bf{\Sigma}}_{\mathrm{t}} \! =  \! {\sigma _0^2{ {\bf{A}}\!( {{\theta _0}} ) {{\bf{T}}_{\mathrm{RF}}} { \sum\limits_{k = 1}^K ({{{\bf{T}}_{{\mathrm{D}},k}} {{\bf{T}}_{{\mathrm{D}},k}^H} }) } {{\bf{T}}_{\mathrm{RF}}^H} {{\bf{A}}^H}\!( {{\theta _0}}) }} $, and ${\bf{\Sigma}}_{\mathrm{cn}} \!\! =  \!\!\sum\limits_{j = 1}^J \! {\sigma _{\mathrm{C}}^2{\bf{A}}\!( {{\theta _j}} ) {{\bf{T}}_{\mathrm{RF}}} { \sum\limits_{k = 1}^K ({{{\bf{T}}_{{\mathrm{D}},k}} {{\bf{T}}_{{\mathrm{D}},k}^H} }) } {{\bf{T}}_{\mathrm{RF}}^H}   {{\bf{A}}^H}\!( {{\theta _j}})}  +  \sigma _{\mathrm{R}}^2{\bf{I}}$.

We note that SCNR is a measurable sensing performance indicator at the receiver, and can be reasonably adapted to different signal characteristics and application scenarios. Moreover, the radar detection probability is known to be a non-decreasing function of the SCNR\footnote{ The target detection probability (PD) increases monotonically with the received SCNR, i.e., $P_{\mathrm{D}}(\gamma) = Q (\sqrt{2{\gamma}},\sqrt{-2 \ln P_{\mathrm{FA}}})$ \cite{9537599,10007634}, where $Q (\cdot,\cdot)$ is the Marcum Q function and $P_{\mathrm{FA}}$ gives the false-alarm probability.}, indicating the relevance of the SCNR in DFRC systems. The measurability and relevance of SCNR offer excellent practicality, leading it to be broadly considered in radar and DFRC systems \cite{10007634}.

\setlength{\parskip}{0.2cm plus4mm minus3mm}
\subsection{Problem Statement}

We holistically design the hybrid transmit beamformers ${\bf{T}}_{{\mathrm{D}},k}$ and ${\bf{T}}_{\mathrm{RF}}$, and the hybrid receive beamformers ${\bf{W}}_{\mathrm{D}}$ and ${\bf{W}}_{\mathrm{RF}}$ under the RS-based HAD architecture, to elevate the downlink sum-rate while ensuring the radar detection performance. This problem is cast as follows:
\begin{subequations}
\begin{align}
    \mathrm{P}_1: \mathop {\max }\limits_{ \mathcal{X} } \
    & R (  {{{\bf{T}}_{\mathrm{RF}}} , \!{{\bf{T}}_{{\mathrm{D}},k}}}, {\bf{U}}_k )  
    \label{eq:problem1A} \\
    \text{s.t.} \ &\varGamma ( {{{\bf{T}}_{\mathrm{RF}}},{{\bf{T}}_{\mathrm{D}}}, {{\bf{W}}_{\mathrm{RF}}},{{\bf{W}}_{\mathrm{D}}} } ) \ge \gamma, \label{eq:problem1B} \\ 
    &\left\| {{{\bf{T}}_{\mathrm{RF}}}{{\bf{T}}_{\mathrm{D}}}} \right\|_F^2 \le P_{\mathrm{T}}, \label{eq:problem1C} \\
    &{{\bf{T}}_{\mathrm{RF}}}( {m,n} )\! \in \! \left\{ {0,{\cal F}} \right\}, \forall m,n, \label{eq:problem1D}  \\
    &{\left\| {{{\bf{T}}_{\mathrm{RF}}}( {m,:} )} \right\|_0} = 1, \forall m, \label{eq:problem1E}  \\
    &{{\bf{W}}_{\mathrm{RF}}}( {\tilde{m},\tilde{n}} )\! \in \! \left\{ {0, {\cal F}} \right\}, \forall \tilde{m},\tilde{n}, \label{eq:problem1F}  \\
    &{\left\| {{{\bf{W}}_{\mathrm{RF}}}( {\tilde{m},:} )} \right\|_0} = 1, \forall \tilde{m}, \label{eq:problem1G} 
\end{align}
\end{subequations}
where $\mathcal{X}=\{ { {{\bf{T}}_{\mathrm{RF}}}, 
{{\bf{T}}_{\mathrm{D},k}},  {\bf{U}}_k ,{{\bf{W}}_{\mathrm{RF}}},{{\bf{W}}_{\mathrm{D}}} } \}$ collects all of the optimization variables. Constraint \eqref{eq:problem1B} ensures the minimum SCNR $\gamma$. Constraint \eqref{eq:problem1C} specifies the total transmit power budget. 
Constraints \eqref{eq:problem1D}-\eqref{eq:problem1G} specify the HBF mapping criteria. More specifically, if the $n$-th RF chain is associated with the $m$-th antenna via a PS, ${{\bf{T}}_{\mathrm{RF}}} \left( {m,n} \right) \in {\cal F}$, $ m = 1,\cdots,{M_\mathrm{T}},n = 1,\cdots,{N_\mathrm{RF}^{\mathrm{t}}} $, contains a nonzero phase, where ${\cal F}$ collects all possible values of a PS; otherwise, ${{\bf{T}}_{\mathrm{RF}}} \left( {m,n} \right) = 0$. Similarly, if the $\tilde n$-th RF chain is connected to the $\tilde m$-th antenna via a PS, ${{\bf{W}}_{\mathrm{RF}}} \left( {\tilde m,\tilde n} \right) \in {\cal F}$, $ \tilde m = 1,\cdots,{M_\mathrm{R}},\tilde n = 1,\cdots,{N_\mathrm{RF}^{\mathrm{r}}} $; otherwise, ${{\bf{W}}_{\mathrm{RF}}} \left( {\tilde m,\tilde n} \right) = 0$. Since no subarrays overlap, the analog beamformer has only one nonzero element per row, i.e., ${\left\| {{{\bf{T}}_{\mathrm{RF}}}\left( {m,:} \right)} \right\|_0} = 1$ for $ m = 1,\cdots,{M_\mathrm{T}} $; ${\left\| {{{\bf{W}}_{\mathrm{RF}}}\left( {\tilde{m},:} \right)} \right\|_0} = 1$ for $ \tilde{m} = 1,\cdots,{M_\mathrm{R}} $.

Problem $\mathrm{P}_1$ poses significant challenges ascribed to non-convex objectives and constraints, and the strong coupling of digital and analog beamformers. The non-convexity makes it difficult to attain a globally optimal solution.
A direct attempt to tackle the coupling effect with block decomposition methods, e.g., block successive upper-bound minimization and inexact flexible parallel algorithm, can also lead to difficulty in ensuring convergence~\cite{9120361}. In what follows, we develop a tractable solution to this problem.

\section{WPDD Algorithm for DFRC with RS Architecture}
\subsection{Reformulation of Problem $\mathrm{P}_1$}

To overcome the difficulty in solving the nonconvex sum-rate problem $\mathrm{P}_1$ with coupled variables, we develop a tractable objective according to the WMMSE method to decrease the complexity of solving problem $\mathrm{P}_1$. As revealed in \textbf{Proposition~\ref{proposition1}}, problem $\mathrm{P}_1$ can be rewritten equivalently as a WMMSE minimization problem \cite{5756489}.

\begin{proposition}   \label{proposition1}
Problem $\mathrm{P}_1$ and the following problem of minimizing weighted sum mean square error (MSE) have the same global optimum:
\begin{subequations} 
\begin{align}
    \mathrm{P}_2: \mathop {\max }\limits_{ \mathcal{X}, \{ {\bf{G}}_k \},\{ {\bf{E}}_k \} }& \sum\limits_{k = 1}^K {( \log \det ({\bf{G}}_k) - \mathrm{tr} ({\bf{G}}_k {\bf{E}}_k )  ) }  \label{eq:problem2A} \\
    \mathrm{s.t.} \  
    &(\mathrm{\ref{eq:problem1B}})-(\mathrm{\ref{eq:problem1G}}),  \nonumber
\end{align}
\end{subequations}
where ${\bf{G}}_k$ denotes the weighting factor for user $k$, ${\bf{U}}_k$ is the receive beamforming for user $k$, and ${\bf{E}}_k$ is the MSE of user $k$, as given by
\begin{align} \label{eq:ek}
    {{\bf{E}}_k} &= \mathbb{E}{\{({\hat{\bf{s}}_k} - {\bf{s}}_k ) ({\hat{\bf{s}}_k} - {\bf{s}}_k )^H \}} \nonumber \\ 
    &= {\bf{I}}_{d_s} -2{\cal R}({ {\bf{U}}_k^H {\bf{H}}_k {{\bf{T}}_{\mathrm{RF}}}{{\bf{T}}_{{\mathrm{D}},k}} }) + {\sigma ^2} {{\bf{U}}_k^H}{{\bf{U}}_k}  \nonumber \\ 
    & \ \ + \sum\limits_{i =1 }^K \!{\bf{U}}_k^H {{\bf{H}}_k{{\bf{T}}_{\mathrm{RF}}}{{\bf{T}}_{{\mathrm{D}},i}} {{\bf{T}}_{{\mathrm{D}},i}^H} {{\bf{T}}_{\mathrm{RF}}^H} { {\bf{H}}_k^H } } {{\bf{U}}_k}.  
\end{align} 

\end{proposition}

\begin{proof}
 Since $ {\bf{G}}_k$ and ${\bf{U}}_k$ appear only in \eqref{eq:problem2A}, the equivalence between problems $\mathrm{P}_2$ and $\mathrm{P}_1$ can be established by substituting their optimal solutions (by first-order optimality condition) into \eqref{eq:problem2A}. For more details, please refer to \cite[Thm 1]{5756489}.
\end{proof}

Although the WMMSE-based transformation makes the objective function easy to handle, the analog and digital beamforming are coupled in constraints \eqref{eq:problem1B} and \eqref{eq:problem1C}.
Moreover, the non-convex SCNR constraint in \eqref{eq:problem1B}, and the $L_0$ constraints in \eqref{eq:problem1E} and \eqref{eq:problem1G} make the problem challenging. Jointly optimizing 
analog and digital beamforming is difficult. We introduce auxiliary variables, ${{\bf{T}}_k} = {{\bf{T}}_{\mathrm{RF}}}{{\bf{T}}_{{\mathrm{D}},k}},\,\forall k$, and ${\bf{W}} = {{\bf{W}}_{\mathrm{RF}}}{{\bf{W}}_{\mathrm{D}}}$.
Problem $\mathrm{P}_2$ is recast as
\begin{subequations}
\begin{align}
    \mathrm{P}_3 : \mathop {\max }\limits_{\mathcal{\tilde{X}}} 
    &\sum\limits_{k = 1}^K {( \log \det ({\bf{G}}_k) - \mathrm{tr} ({\bf{G}}_k {\bf{E}}_k )  ) }   \label{eq:problem3A} \\
    \mathrm{s.t.} \ & \frac{ \mathrm{tr} ({{{\bf{W}}^H}{{\bf{\Sigma}}_{\mathrm{t}}} {\bf{W}} } )}{ \mathrm{tr} ({{\bf{W}}^H}{\bf{\Sigma}}_{\mathrm{cn}} {\bf{W}} )} \ge \gamma, \label{eq:problem3B}\\
    & \sum\limits_{k = 1}^K \left\| {{\bf{T}}_k} \right\|^2 \le P_{\mathrm{T}}, \label{eq:problem3C} \\
    & {{\bf{T}}_k} = {{\bf{T}}_{\mathrm{RF}}}{{\bf{T}}_{{\mathrm{D}},k}},\forall k,   \label{eq:problem3D} \\  
    & {\bf{W}} = {{\bf{W}}_{\mathrm{RF}}}{{\bf{W}}_{\mathrm{D}}},   \label{eq:problem3E} \\
    &\mathrm{\eqref{eq:problem1D}}-\mathrm{\eqref{eq:problem1G}}, \nonumber 
\end{align}
\end{subequations}
where $\mathcal{\tilde{X}} \!=\! \{\! { {{\bf{T}}_{\mathrm{RF}}},{{\bf{T}}_{\mathrm{D},k}} ,  {\bf{U}}_k , {{\bf{W}}_{\mathrm{RF}}},{{\bf{W}}_{\mathrm{D}}},  {\bf{G}}_k , {{\bf{T}}_k} , {\bf{W}}  } \!\}$.

\subsection{PDD-Based Framework for Problem $\mathrm{P}_3$}

We resort to the PDD to solve $\mathrm{P}_3$, as problem $\mathrm{P}_3$ has the form of the problem ($\mathcal{P}$) in \cite{9120361}. 
The PDD method integrates the penalty method, the augmented Lagrangian method, and the alternating direction method of multipliers (ADMM) and delivers a general algorithmic framework for minimizing a non-convex and non-smooth function while dealing with non-convex coupling constraints. The PDD comprises double-loop iterations, with the outer loops updating the dual variables and penalty parameters, and the inner loops executing the augmented Lagrangian method. The PDD suits problem $\mathrm{P}_3$ with a differentiable objective function.  

Through the incorporation of two penalty terms for constraint \eqref{eq:problem3D} and \eqref{eq:problem3E} into \eqref{eq:problem3A}, the inner loop can be formulated as 
\begin{subequations}
\begin{align}
    \mathrm{P}_4:\!\mathop {\min }\limits_{\mathcal{\tilde{X}}} 
    & \sum\limits_{k = 1}^K \{ {  \mathrm{tr} ( {\bf{G}}_k {\bf{E}}_k ) \!-\! \log \det ( {\bf{G}}_k ) } \nonumber \\
    & \ \ + {\frac{1}{2\rho } { \left \|   {{{\bf{T}}_k} \! - \!{{\bf{T}}_{\mathrm{RF}}}{{\bf{T}}_{{\mathrm{D}},k}} \! + \! \rho {{\bf{D}}_k}}  \right \| }^2} \}  \nonumber \\
    & \ \ +  {\frac{1}{2\rho }\! { \left \|  {{\bf{W}} \! - \!{{\bf{W}}_{\mathrm{RF}}}\!{{\bf{W}}_{\mathrm{D}}} \! + \! \rho {\bf{\tilde{D}}}}  \right \| }^2}   \label{eq:problem4A} \\
    \mathrm{s.t.} \  & \mathrm{\eqref{eq:problem1D}}- \mathrm{\eqref{eq:problem1G},\, \eqref{eq:problem3B}}-\mathrm{\eqref{eq:problem3C}},    \nonumber 
\end{align}
\end{subequations}
where $\rho$ denotes the penalty parameter; ${{\bf{D}}_k}$ and ${\bf{\tilde{D}}}$ are the dual variables of constraints \eqref{eq:problem3D} and \eqref{eq:problem3E}, respectively. $\rho$, ${{\bf{D}}_k}, \forall k$, and ${\bf{\tilde{D}}}$ are fixed throughout the inner-loop iterations.

Next, we solve problem $\mathrm{P}_4$ utilizing the PDD framework, where BCD is used to solve problem $\mathrm{P}_4$.

\noindent \emph{1) Equivalent FD Sensing Receive Beamformer ${\bf{W}}$}: 

By letting ${{\bf{t}}_k}= {\mathrm{vec}} ({{\bf{T}}_k})$, ${\bf{t}} = [{\bf{t}}_1^T,\cdots,{\bf{t}}_K^T]^T$, and ${\bf{w}}= {\mathrm{vec}} ({\bf{W}})$, the SCNR in \eqref{Rsnr} can be rewritten as
\begin{align}   \label{Rsnr2}
   & \varGamma ( {{{\bf{T}}_{\mathrm{RF}}},{{\bf{T}}_{\mathrm{D}}},{\bf{W}}_{{\mathrm{RF}}},{\bf{W}}_{{\mathrm{D}}} } )  \!=\! \frac{ {\sigma _0^2} {\left| {\bf{w}}^H \tilde{\bf{A}}\left( {{\theta _0}} \right) {\bf{t}} \right|^2}}{ {{{\bf{w}}^H}{\tilde{\bf{\Sigma}}_{\mathrm{cn}}} {{\bf{w}} }}  }, 
\end{align}
where $\tilde{\bf{A}}(  \theta ) \! = \! \bf{I} \otimes {\bf{A}}(  \theta )$, $\tilde{\bf{\Sigma}}_{\mathrm{cn}} \!\! =  \!\!\sum\limits_{j = 1}^J \! {\sigma _{\mathrm{C}}^2\tilde{\bf{A}}\!( {{\theta _j}} ) {\bf{t}} {\bf{t}}^H  {\tilde{\bf{A}}}^H\!( {{\theta _j}})}  \!+\!  \sigma _{\mathrm{R}}^2{\bf{I}}$.  

When the transmit beamforming is given and fixed, we employ the well-known minimum variance distortionless response (MVDR) beamformer \cite{8239836} to obtain the maximum SCNR at the MIMO radar, as given by
\begin{equation} \label{w}
{\bf{w}} = \frac{ {{\tilde{\bf{\Sigma}}_{\mathrm{cn}}^{-1}} \tilde{\bf{A}}( {{\theta _0}} ){\bf{t}}} } { {{\bf{t}}^H} {\tilde{\bf{A}}}^H( {{\theta _0}} ){\tilde{\bf{\Sigma}}_{\mathrm{cn}}^{-1}} \tilde{\bf{A}}( {{\theta _0}} ){\bf{t}} },    
\end{equation}
then $\bf{W}$ can be obtained by reshaping $\bf{w}$. According to this optimal radar receive beamformer $\bf{W}$, the SCNR in \eqref{Rsnr2} is further expressed as
\begin{equation}  \label{gamma}
\tilde{\varGamma} ( {{{\bf{T}}_{\mathrm{RF}}},{{\bf{T}}_{\mathrm{D},k}}} ) = \sum\limits_{k = 1}^K \mathrm{tr} \left({\bf{T}}_{{\mathrm{D}},k}^H {\bf{T}}_{\mathrm{RF}}^H {\bf{\Phi }} {\bf{T}}_{\mathrm{RF}}{\bf{T}}_{{\mathrm{D}},k}\right), 
\end{equation}
where ${\bf{\Phi }} = \sigma _0^2{{\bf{A}}^H}( {{\theta _0}} ){{\bf{\Sigma}}_{\mathrm{cn}}^{-1}}{\bf{A}}( {{\theta _0}} ) $.

\noindent \emph{2) Equivalent FD Transmit Beamformer
${{\bf{T}}_k}$}: \par

With the rest of variables fixed, the subproblem of problem $\mathrm{P}_4$ for solving the auxiliary variable ${{\bf{T}}_k}$ is written as
\begin{subequations}  \label{problem:tk}
\begin{align}
    \mathop {\min }\limits_{\left\{ {{\bf{T}}_k} \right\}} &\sum\limits_{k = 1}^K ( \mathrm{tr} ( {\bf{G}}_k {\bf{E}}_k )  \!+\!
    {\frac{1}{2\rho } { \left \|   {{{\bf{T}}_k} \! - \!{{\bf{T}}_{\mathrm{RF}}}{{\bf{T}}_{{\mathrm{D}},k}} \! + \! \rho {{\bf{D}}_k}}  \right \| }^2} ) \label{eq:problemtkA} \\ 
    \mathrm{s.t.} \ 
    &\sum\limits_{k = 1}^K \mathrm{tr}\left({{\bf{T}}_k^H{\bf{\Phi }}{{\bf{T}}_k}}\right) \ge \gamma,  \label{eq:problemtkB} \\
    &\sum\limits_{k = 1}^K {\left\| {{\bf{T}}_k} \right\|^2}  \le {P_{\mathrm{T}}}, \label{eq:problemtkC} 
\end{align}
\end{subequations}
where ${{\bf{E}}_k}$ is obtained by plugging \eqref{eq:problem3D} into \eqref{eq:ek}, as given by
\begin{align}  \label{eq:ek2}
    {{\bf{E}}_k} & =  {\bf{I}}_{d_s} -2{\cal R}({ {\bf{U}}_k^H {\bf{H}}_k {{\bf{T}}_k} }) + {\sigma ^2} {{\bf{U}}_k^H}{{\bf{U}}_k}  \nonumber \\ 
    & \ \ + \sum\limits_{i =1 }^K \!{\bf{U}}_k^H {{\bf{H}}_k {{\bf{T}}_i}{{\bf{T}}_i^H} { {\bf{H}}_k^H } } {{\bf{U}}_k}.
\end{align}

Notice that \eqref{problem:tk} provides a non-convex quadratic constraint quadratic programming (QCQP). Semi-definite relaxation (SDR) is often applied to resolve such a problem. 
However, SDR conducts a relaxation of rank-one constraints, resulting in an additional process like Gaussian randomization to restore the rank-one constraints on obtained solutions \cite{5447068}. Extra computational complexity occurs.

We propose to transform  \eqref{problem:tk} into a series of SOCPs \cite{9772613}, using viable substitution and successive convex approximation (SCA) techniques and solve the SOCPs recursively to obtain a sub-optimal solution to \eqref{problem:tk}.  Let ${\bf{Z}}_k = {\bf{T}}_k + \rho {\bf{D}}_k$ and  ${\bf{P}}_k = {\bf{H}}_k^H {\bf{U}}_k {{\bf{G}}_k}{{\bf{U}}_k^H} { {\bf{H}}_k }$. We introduce slack variables ${a_{k,i}} > 0$ and ${b_k} > 0$, $ i,k\in \{1,\cdots,K \} $ such that ${ {{\bf{t}}_i^H} ({\bf{I}} \otimes {\bf{P}}_k^H) {{\bf{t}}_i}  }  \le {a_{k,i}}$ and $ { \left \| {{{\bf{Z}}_k} \! - \!{{\bf{T}}_{\mathrm{RF}}}{{\bf{T}}_{{\mathrm{D}},k}}}  \right \| }^2 \le {b_k} $, which can be further rewritten in the following second-order cone (SOC) constraints: 
\begin{equation}\label{eq:aki}
    \left\| {\left[ {\begin{array}{*{20}{c}}
    {{\bf{t}}_i^H} ({\bf{I}} \otimes {\bf{P}}_k^H)^{1/2}  \\
    {\frac{{{a_{k,i}} - 1}}{2}}
    \end{array}} \right]} \right\| \le \frac{{{a_{k,i}} + 1}}{2},  \forall i,k;
\end{equation}
\begin{equation}\label{eq:bk}
    \left\| {\left[ {\begin{array}{*{20}{c}}
    {{{\bf{Z}}_k} \! - \!{{\bf{T}}_{\mathrm{RF}}}{{\bf{T}}_{{\mathrm{D}},k}}} \\
    {\frac{{{b_k} - 1}}{2}}
    \end{array}} \right]} \right\| \le \frac{{{b_k} + 1}}{2}, \forall k. 
\end{equation}
In the non-convex SCNR constraint \eqref{eq:problemtkB}, we can rewrite $\mathrm{tr}\left({{\bf{T}}_k^H{\bf{\Phi }}{{\bf{T}}_k}}\right)=  {{\bf{t}}_k^H} ({\bf{I}} \otimes {\bf{\Phi}}^H) {{\bf{t}}_k}$. Then, we resort to the SCA technique to replace this non-convex term with its linear lower bound (e.g., the first-order expansion), i.e., 
\begin{equation}  \label{eq:gamma}
  {\bar{\varGamma}}_k =  {2{\cal R} ( {{\bf{\bar t}}_k^H \bar{{\bf{\Phi }}}{{\bf{t}}_k}} ) - {\bf{\bar t}}_k^H \bar{\bf{\Phi}}{{{\bf{\bar t}}}_k}},
\end{equation}
where $\bar{\bf{\Phi}}= {\bf{I}} \otimes {\bf{\Phi}}^H$, and ${{{\bf{\bar t}}}_k}$ is the value of ${\bf{t}}_k$ in the previous SCA iteration.

Thus, the approximate convex problem of \eqref{problem:tk} is 
\begin{subequations} \label{problem:SOCP}
\begin{align}
    \mathop {\min}\limits_{ \left\{  {{\bf{t}}_k}, {a_{k,i}}, {b_k}  \right\}} &\sum\limits_{k = 1}^K \!{\left( \! {{\sum\limits_{i = 1}^K {a_{k,i}} } \!-\! {2{\cal R}( \mathrm{vec}^H ({\bf{G}}_k {\bf{U}}_k^H {\bf{H}}_k )^H {{\bf{t}}_k} )} \! + \! \frac{1}{2\rho} {b_k} } \! \right)}   \\
    \mathrm{s.t.} \ &\sum\limits_{k = 1}^K {2{\cal R} ( {{\bf{\bar t}}_k^H \bar{{\bf{\Phi }}}{{\bf{t}}_k}} ) - {\bf{\bar t}}_k^H \bar{\bf{\Phi}}{{{\bf{\bar t}}}_k}}  \ge \gamma ,   \\
    &{\left\| {\bf{t}} \right\|}  \le \sqrt{P_{\mathrm{T}}},\\   &\eqref{eq:aki}-\eqref{eq:bk}.\nonumber
\end{align}
\end{subequations}
We can use off-the-shelf toolboxes, e.g., CVX \cite{grant2014cvx}, to efficiently solve \eqref{problem:SOCP}. By recursively solving \eqref{problem:SOCP}, we can obtain at least a local optimum of \eqref{problem:tk}.

Through the use of the PDD framework, we can decouple hybrid transmit and receive beamforming for DFRC into the subproblems of optimizing the equivalent sensing receive beamformer ${\bf{W}} $ and the equivalent DFRC transmit beamformers ${{\bf{T}}_k},\,\forall k$, which can be solved sequentially. The inter-user interference in the considered system leads to non-convexity in the objective function in \eqref{eq:problemtkA}. The clutter interference leads to non-convexity in the SCNR constraint in~\eqref{eq:problemtkB}. To deal with this, variable substitution and SCA are conducted to convexify both the objective function and the SCNR constrain in the subproblem of ${{\bf{T}}_k},\,\forall k$; see \eqref{eq:aki}, \eqref{eq:bk} and \eqref{eq:gamma}. As a result, the subproblem can be solved using off-the-peg CVX toolkits.

\vspace{0.2cm}
\noindent
{\em Remark}: 
PDD is a general algorithmic framework for solving problems with multiplicatively coupled variables. It has been applied to optimize hybrid transmit and receive beamforming for a single user in a MIMO channel, e.g., in~\cite{9110865}. By contrast, the studied problem of multi-user mMIMO can undergo non-negligible inter-user interference; see \eqref{ysignal2}. Moreover, our studied system suffers from interference from clutters to the radar sensing signals received by the BS; see \eqref{rsignal2}. This differs substantially from the work presented in~\cite{9110865}, where no such interference exists. In this sense, our studied problem is distinct from~\cite{9110865} in objectives and constraints.

\noindent \emph{3) Digital Receive Beamformer ${\bf{U}}_k$ for User $k$}: \par
We optimize {$\{{\bf{U}}_k \}$} by fixing all other variables, leading to an unconstrained problem:
\begin{equation} \label{problem:uk}
\mathop {\min } \limits_{ {\bf{U}}_k } \sum\limits_{k = 1}^K { \mathrm{tr} ( {\bf{G}}_k {\bf{E}}_k )} .
\end{equation}
By substituting \eqref{eq:ek2} into \eqref{problem:uk} and then assessing the first-order derivative of \eqref{problem:uk}, the optimum receive combiner ${\bf{U}}_k$ for user $k$ is obtained in closed form as follows:
\begin{equation}
{\bf{U}}_k = \left ({\sum\limits_{i =1 }^K {{\bf{H}}_k}{{\bf{T}}_i}{{\bf{T}}_i^H}{{\bf{H}}_k^H} + \sigma^2 {\bf{I}}_{M_{\mathrm{U}}} } \right )^{-1} \! {{\bf{H}}_k {{\bf{T}}_k}}, \forall k. \label{eq:uk}
\end{equation}
\par

\noindent \emph{4) WMMSE Weighting Matrix ${\bf{G}}_k$}: \par
By substituting \eqref{eq:uk} into \eqref{eq:ek2}, the mean square error $ {{\bf{E}}_k} $ is given by $( {{\bf{I}} - {\bf{U}}_k^H {\bf{H}}_k {{\bf{T}}_k}} )$. Then, assessing the first-order optimality condition of problem $\mathrm{P}_4$ concerning ${\bf{G}}_k$, the closed-form solution of ${\bf{G}}_k$ is
\begin{equation} 
{\bf{G}}_k = {( {{\bf{E}}_k} )^{ - 1}} = ( {{\bf{I}} - {\bf{U}}_k^H {\bf{H}}_k {{\bf{T}}_k}} )^{-1},\forall k. 
 \label{eq:wk}
\end{equation}

\noindent \emph{5) Analog Transmit Beamformer ${{\bf{T}}_{\mathrm{RF}}}$}: \par
Recall that ${\bf{Z}}_k = {\bf{T}}_k + \rho {\bf{D}}_k$, and the analog transmit beamformer ${{\bf{T}}_{\mathrm{RF}}}$ is only in the second term of \eqref{eq:problem4A}. Therefore, the subproblem concerning ${{\bf{T}}_{\mathrm{RF}}}$ can be given by
\begin{equation}  \label{problem:TRF}
\begin{aligned}
    \mathop{\min } \limits_{{\bf{T}}_{\mathrm{RF}}} &\sum\limits_{k = 1}^K { {\left\| {{\bf{Z}}_k} - {{\bf{T}}_{\mathrm{RF}}}{{\bf{T}}_{{\mathrm{D}},k}} \right\|^2} }    \\
    \mathrm{s.t.} \ 
    &{{\bf{T}}_{\mathrm{RF}}}\left( {m,n} \right) \in \left\{ {0,{\cal F}} \right\}, \forall m,n,  \\
    &{\left\| {{{\bf{T}}_{\mathrm{RF}}}\left( {m,:} \right)} \right\|_0} = 1, \forall m. 
\end{aligned}
\end{equation}
Let ${\bf{Z}}  = \left[ {{{\bf{Z}}_1},\cdots,{{\bf{Z}}_K}} \right]$. Then, the objective $\sum\nolimits_{k = 1}^K { {\left\| {{\bf{Z}}_k} - {{\bf{T}}_{\mathrm{RF}}}{{\bf{T}}_{{\mathrm{D}},k}} \right\|^2} }$ can be rewritten as $\left\| {{\bf{Z}} - {{\bf{T}}_{\mathrm{RF}}}{{\bf{T}}_{\mathrm{D}}}} \right\|^2 $. The constraints in \eqref{problem:TRF} indicate there is only one non-zero element with a constant magnitude in every row of ${{\bf{T}}_{\mathrm{RF}}}$, allowing \eqref{problem:TRF} to be solved row-by-row. \par

Let ${{\cal S}_n},n = 1,\cdots,{N_{\mathrm{RF}}^{\mathrm{t}}} $  denote the transmit antennas linked to the $n$-th RF chain. Because each of the antennas is linked to only an RF chain and the subarrays do not intersect each other, we can have
\begin{equation}\label{eq:DCcons1}
\cup _{n = 1}^{N_{\mathrm{RF}}^{\mathrm{t}}} {{\cal S}_n} = \left\{ {1,\cdots,{M_{\mathrm{T}}}} \right\}, \forall n; 
\end{equation}
\begin{equation}\label{eq:DCcons2}
{{\cal S}_p} \cap {{\cal S}_q} = \emptyset ,\forall p \ne q. 
\end{equation}
Based on \eqref{eq:DCcons1} and \eqref{eq:DCcons2}, we establish the following proposition to solve problem \eqref{problem:TRF}.

\begin{proposition}\label{Proposition2}
Define ${\bf{Z}}( {m,:} ){{\bf{T}}_{\mathrm{D}}}{( {n,:} )^H} = \left| {{\zeta _{m,n}}} \right|{e^{j{\tilde{\phi} _{m,n}}}} $ for $ m \in {{\cal S}_n} $ and $ n =  1,\cdots,{N_{\mathrm{RF}}^{\mathrm{t}}} $. Problem \eqref{problem:TRF} is rewritten as
\begin{equation}  \label{problem:DSmap}
\begin{aligned}
    \mathop {\max}\limits_{\left\{ {\varphi _{m,n}} \right\} \left\{ {{\cal S}_n} \right\}} &\sum\limits_{n = 1}^{N_{\mathrm{RF}}^{\mathrm{t}}}{\sum\limits_{{\forall m} \in {\cal S}_n} {\left| {\zeta _{m,n}} \right| \cos ( {{\tilde{\phi} _{m,n}} - {\varphi _{m,n}}} )} } \\
    \mathrm{s.t.} \ & \text{\eqref{eq:DCcons1},  \eqref{eq:DCcons2}},
\end{aligned}
\end{equation}
where ${\varphi _{m,n}}$ stands for the phase of the $( {m,n} )$-th element of ${{\bf{T}}_{\mathrm{RF}}}$, i.e., $ {{\bf{T}}_{\mathrm{RF}}}( {m,n} ) = {e^{j{\varphi _{m,n}}}}$.   
\end{proposition}
\begin{proof}
See \textbf{Appendix \ref{appendix_P2}}.
\end{proof}

Clearly, the objective of problem \eqref{problem:DSmap} is  maximized when ${\varphi _{m,n}} = {\tilde{\phi} _{m,n}}$. Therefore, for each antenna $m$, we map it to the $n$-th RF chain maximizing $\left| {\zeta _{m,n}} \right|\cos ( {{\tilde{\phi} _{m,n}} - {\varphi _{m,n}}} )$:
\begin{equation}
n_m^{\star} = \arg \mathop {\max }\limits_n \left| {\zeta _{m,n}} \right|\cos ( {{\tilde{\phi} _{m,n}} - {\varphi _{m,n}}} ), 
\label{RFset}
\end{equation}
and refresh the subset of transmit antennas for the $n$-th RF chain: ${{\cal S}_{n^{\star}}} = {{\cal S}_{n^{\star}}} \cup \left\{ m \right\}$.\par
The optimal subarray mapping scheme is obtained by performing \eqref{RFset} for each antenna. Then the corresponding optimal analog beamformer is obtained as 
\begin{equation}  \label{TRF}
{\bf{T}}_{\mathrm{RF}} {( {m,n} )} \!=\! 
\begin{cases}
{e^{j{\varphi _{m,n}}}}, & \mathrm{if} \ m \!\in \!{{\cal S}_{n^{\star}}}, \,n \!=\!  1,\!\cdots,\!{N_{\mathrm{RF}}^{\mathrm{t}}};  \\
{0,} &{\mathrm{otherwise.}}
\end{cases}
\end{equation}
Based on the sparsity in the reconfigurable connection constraint \eqref{eq:problem1D}, we guarantee the optimality of the analog beamformer ${\bf{T}}_{\mathrm{RF}}$ since rows are enumerated one after another.

\noindent \emph{6) Analog Sensing Receive Beamformer ${\bf{W}}_{\mathrm{RF}}$}: \par
Define $ {\bf{Q}} ={ {\bf{W}} \! + \! \rho {\bf{\tilde{D}}} } $. The subproblem concerning ${{\bf{W}}_{\mathrm{RF}}}$ is
\begin{equation}  \label{problem:WRF}
\begin{aligned}
\mathop{\min } \limits_{{\bf{W}}_{\mathrm{RF}}} &  {\left\| {\bf{Q}} - {{\bf{W}}_{\mathrm{RF}}}{{\bf{W}}_{\mathrm{D}}} \right\|^2}    \\
    \mathrm{s.t.} \ 
    &{{\bf{W}}_{\mathrm{RF}}}( {\tilde{m},\tilde{n}} )\! \in \! \left\{ {0, {\cal F}} \right\}, \forall \tilde{m},\tilde{n},  \\
    &{\left\| {{{\bf{W}}_{\mathrm{RF}}}( {\tilde{m},:} )} \right\|_0} = 1, \forall \tilde{m}. 
\end{aligned}
\end{equation}
Since problem \eqref{problem:WRF} has the same form as problem \eqref{problem:TRF}, we can solve the radar receive analog beamforming ${\bf{W}}_{\mathrm{RF}}$ in the same way as described in \textbf{Proposition~\ref{Proposition2}}, as given by 
\begin{equation}  \label{problem:DSmap2}
\begin{aligned}
    \mathop{\max}\limits_{ \left\{ { \bar{\varphi} _{\tilde{m},\tilde{n}}} \right\},\left\{ {\bar{\cal S}}_{\tilde{n}} \right\}  }  &\sum\limits_{{\tilde{n}} = 1}^{N_{\mathrm{RF}}^{\mathrm{r}}}{\sum\limits_{{\forall {\tilde{m}}} \in {\bar{\cal S}}_{\tilde{n}}} {\left| {\bar{\zeta} _{{\tilde{m}},{\tilde{n}}}} \right| \cos ( {{\bar{\phi} _{{\tilde{m}},{\tilde{n}}}} - {\bar{\varphi} _{{\tilde{m}},{\tilde{n}}}}} )} }   \\
    \mathrm{s.t.} \ & \cup _{\tilde{n} = 1}^{N_{\mathrm{RF}}^{\mathrm{r}}} \bar{\cal S}_{\tilde{n}} = \left\{ {1,\cdots,{M_{\mathrm{R}}}} \right\}, \forall \tilde{n}, \\
   &\bar{{\cal S}_p} \cap {\bar{\cal S}_q} = \emptyset ,\forall p \ne q,
\end{aligned}
\end{equation}
where $\bar{\cal S}_{\tilde{n}}$, ${\tilde{n}}= 1,\cdots,{N_{\mathrm{RF}}^{\mathrm{r}}}$ stands for the set of receive antennas linked to the ${\tilde{n}}$-th RF chain; ${\bar{\zeta} _{{\tilde{m}},{\tilde{n}}}}$ and ${\bar{\phi} _{{\tilde{m}},{\tilde{n}}}}$ stand for the amplitude and phase of ${\bf{Q}}( {{\tilde{m}},:} ){{\bf{W}}_{\mathrm{D}}}{( {{\tilde{n}},:} )^H}$, respectively; ${\bar{\varphi} _{{\tilde{m}},{\tilde{n}}}}$ is the phase of the $( {{\tilde{m}},{\tilde{n}}} )$-th element of ${{\bf{W}}_{\mathrm{RF}}}$. The corresponding optimal receive analog beamformer is given by
\begin{equation}  \label{WRF}
{\bf{W}}_{\mathrm{RF}} {( {{\tilde{m}},{\tilde{n}}} )} \!=\! 
\begin{cases}
{e^{j{\bar{\varphi} _{{\tilde{m}},{\tilde{n}}}}}}, & \!\mathrm{if} \ {\tilde{m}} \!\in \!{\bar{\cal S}_{{\tilde{n}}^{\star}}}, \,{\tilde{n}} \!=\!  1,\!\cdots,\!{N_{\mathrm{RF}}^{\mathrm{r}}};  \\
{0,} &\!{\mathrm{otherwise.}}
\end{cases}
\end{equation}
Here, ${\bar{\cal S}_{{\tilde{n}}^{\star}}}$ is the set of optimal antennas connected to the ${\tilde{n}}$-th RF chain obtained from \eqref{problem:DSmap2}.

\noindent \emph{7) Digital Transmit Beamformer ${\bf{T}}_{{\mathrm{D}},k}$ and Sensing Receive Beamformer ${\bf{W}}_{\mathrm{D}}$}: 

The subproblem concerning the digital transmit beamforming ${\bf{T}}_{{\mathrm{D}},k}$ is an unconstrained least squares (LS) problem: 
\begin{equation} \label{problem:TDk}
\mathop{\min } \limits_{{\bf{T}}_{{\mathrm{D}},k}} \sum\limits_{k = 1}^K { {\left\| {{\bf{Z}}_k} - {{\bf{T}}_{\mathrm{RF}}}{{\bf{T}}_{{\mathrm{D}},k}} \right\|^2} }. 
\end{equation}
Subsequently, the optimal solution to \eqref{problem:TDk} is obtained as
\begin{equation}   \label{eq:TDk}
{{\bf{T}}_{{\mathrm{D}},k}} = {\bf{T}}_{\mathrm{RF}}^\dag {{\bf{Z}}_k}, \forall k.
\end{equation}

The subproblem concerning the radar digital receive beamforming ${\bf{W}}_{\mathrm{D}}$ can be formulated in the same way as \eqref{problem:TDk}. With reference to \eqref{eq:TDk}, the optimal ${\bf{W}}_{\mathrm{D}}$ can be given by:
\begin{equation}   \label{eq:WD}
{{\bf{W}}_{\mathrm{D}}} = {\bf{W}}_{\mathrm{RF}}^\dag {\bf{Q}}.
\end{equation}

\setParDef

\subsection{Summary of the WPDD Algorithm}

\textbf{Algorithm \ref{alg1}} summarizes the new WPDD algorithm for designing the HBF in the considered DFRC system by solving problem $\mathrm{P}_1$. Let ${\bf{D}}= \left[ {\bf{D}}_1,\cdots,{\bf{D}}_K \right]$. 
In the outer-loop iteration, we refresh the dual variable sets $\{ 
{\bf{D}},\tilde{\bf{D}} \}$  
and the penalty parameter $\rho$, based on the following constraint violation condition: 
\begin{equation}
 h({\mathcal{\tilde{X}}} ) \!\!=\!\! \max({\left\| {{{\bf{T}}_k} \!\!- \!\!{{\bf{T}}_{\mathrm{RF}}}{{\bf{T}}_{{\mathrm{D}},k}}} \right\|, \left\| {{\bf{W}} \!\!- \!\!{{\bf{W}}_{\mathrm{RF}}}{{\bf{W}}_{\mathrm{D}}}} \right\|}). 
\end{equation}
If $ h({\mathcal{\tilde{X}}}) 
\leq \eta$, we update the dual variables by
\begin{align}   
     {\bf{D}}_k &\leftarrow {\bf{D}}_k  +  \frac{1}{\rho}{( {{\bf{T}}_k - {\bf{T}}_{\mathrm{RF}}{\bf{T}}_{\mathrm{D},k}} )}, \forall k, \label{D1} \\ 
     \tilde{\bf{D}} &\leftarrow \tilde{\bf{D}}  +  \frac{1}{\rho}{( {{\bf{W}} - {\bf{W}}_{\mathrm{RF}}{\bf{W}}_{\mathrm{D}}} )}.\label{D2}
\end{align} 
If $ h({\mathcal{\tilde{X}}} ) > \eta $, we reduce the penalty factor $\rho$ by $\rho \leftarrow c \rho$, with $0<c<1$. 
\textbf{Algorithm \ref{alg1}} terminates if $ h({\mathcal{\tilde{X}}} ) 
< \varepsilon _2  $, where $\varepsilon _2 $ is a prespecified threshold. The WPDD stabilizes at a stationary point of problem $\mathrm{P}_1$ \cite{8332507}, as discussed later in Section IV-A.\par


\begin{algorithm}[t]
    \renewcommand{\algorithmicrequire}{\textbf{Initialize:}}
	\renewcommand{\algorithmicensure}{\textbf{Output:}}
	\caption{The proposed WPDD algorithm} \label{alg1}
    \begin{algorithmic}[1] 
        \REQUIRE  ${I_{\rm{max}}}$, $\varepsilon _1$, $\varepsilon _2 $, ${\bf{T}}_{\mathrm{RF}}^{(0)}$, ${\bf{T}}_{\mathrm{D},k}^{(0)}$, ${\bf{U}}_k^{(0)}$, ${\bf{W}}_{\mathrm{RF}}^{(0)}$, ${\bf{W}}_{\mathrm{D}}^{(0)}$, ${\rho ^{(0)}}>0$, $0 < c < 1 $ 
        
        \REPEAT
            \REPEAT
                \STATE Update $\bf{W}$ by \eqref{w}, ${{\bf{T}}_k}$ by \eqref{problem:SOCP}, ${\bf{U}}_k$ by \eqref{eq:uk}, ${\bf{G}}_k$ by~\eqref{eq:wk}, ${\bf{T}}_{\mathrm{RF}}$ by \eqref{TRF}, ${\bf{W}}_{\mathrm{RF}}$ by \eqref{WRF}, $ {\bf{T}}_{{\mathrm{D}},k} $ by \eqref{eq:TDk}, $ {\bf{W}}_{\mathrm{D}} $ by \eqref{eq:WD}
            \UNTIL the convergence with $\varepsilon_1$ or ${I_{\rm{max}}}$ inner iterations
            
            \IF {$ h({\mathcal{\tilde{X}}})  \le \eta  $}
            \STATE Update ${\bf{D}}_k$ by \eqref{D1} and $\tilde{\bf{D}}$ by \eqref{D2}
            \ELSE
            \STATE Update $\rho$ by $\rho \leftarrow c \rho$
            \ENDIF    \\
        ${\eta \leftarrow 0.8  {h({\mathcal{\tilde{X}}}) }}$
    \UNTIL ${h({\mathcal{\tilde{X}}}) } < \varepsilon _2 $ 

        \ENSURE ${\bf{T}}_{\mathrm{RF}},{\bf{T}}_{\mathrm{D},k}$, ${\bf{U}}_k$, ${\bf{W}}_{\mathrm{RF}}$, ${\bf{W}}_{\mathrm{D}}$
    \end{algorithmic}
\end{algorithm}

\subsection{Adaptation Under PC Architectures}

A PC configuration represents a specific instance of RS connection, where every RF chain is persistently hardwired to a set of antennas. Specifically, the analog beamformer ${\bf{T}}_{\mathrm{RF}}$ and ${\bf{W}}_{\mathrm{RF}}$ are block diagonal matrices. For example, every block of transmit analog beamforming ${\bf{T}}_{\mathrm{RF}}$ has an $M$-dimensional vector with a constant-modulus (CM) value, i.e.,
\begin{equation}  \label{eq:PC}
    {{\bf{T}}_{\mathrm{RF}}} ={\rm{blkdiag}} \{{{{\bf{p}}_1},\cdots,{{\bf{p}}_{N_{\mathrm{RF}}^{\mathrm{t}}}}}\} \in {\mathbb{C}}^{{{M_{\mathrm{T}}} \times {N_{\mathrm{RF}}^{\mathrm{t}}}}},
\end{equation}
where ${{\bf{p}}_i} \in {\mathbb{C}}^{M \times 1} $ for $ i = 1,\cdots,{N_{\mathrm{RF}}^{\mathrm{t}}}$, $M = {{M_{\mathrm{T}}} \mathord{\left/ {\vphantom {{{M_{\mathrm{T}}}} {N_{\mathrm{RF}}^{\mathrm{t}}}}} \right. \kern-\nulldelimiterspace} {N_{\mathrm{RF}}^{\mathrm{t}}}}$, and $ \left| {{{\bf{p}}_i}( \nu )} \right| = 1$ for $\nu = 1,\cdots,M.$\par

By replacing the reconfigurable connection constraints \eqref{eq:problem1D} and \eqref{eq:problem1E} in problem $\mathrm{P}_1$ with \eqref{eq:PC}, we can obtain the DFRC HBF design under the PC architecture. Fortunately, \textbf{Algorithm~\ref{alg1}} is still applicable to this problem, except that the subproblems for updating the digital and analog beamformers, ${{\bf{T}}_{\mathrm{RF}}}$ and ${{\bf{T}}_{\mathrm{D},k}}$ differ, that is, Step 3 in \textbf{Algorithm \ref{alg1}}, becomes different, as described below.\par

The subproblem of ${{\bf{T}}_{\mathrm{RF}}}$, i.e., \eqref{problem:TRF}, can be rewritten as
\begin{align}
    &\mathop{\min} \limits_{{\bf{T}}_{\mathrm{RF}}} \left\| {{\bf{Z}} - {{\bf{T}}_{\mathrm{RF}}}{{\bf{T}}_{\mathrm{D}}}} \right\|^2 \nonumber \\
    = &\mathop{\min} \limits_{{\varphi _{m,n}}} \sum\limits_{\forall n} {\sum\limits_{\forall m} {\left\| {{\bf{Z}}( {m,:} ) - {e^{j{\varphi _{m,n}} }}{{\bf{T}}_{\mathrm{D}}}( {n,:} )} \right\|^2} }. 
\end{align}
Here, $m =  1,\cdots,{M_{\mathrm{T}}} $, and $n = 1,\cdots,\left\lceil {m\frac{N_{\mathrm{RF}}^{\mathrm{t}}}{M_{\mathrm{T}}}} \right\rceil  $. 
Then, the phase of the $(m,n)$-th element of ${\bf{T}}_{\mathrm{RF}}$ is given by
\begin{equation} \label{eq:TRF2}
    {\varphi _{m,n}} = \arg \left( {{\bf{Z}}\left( {m,:} \right){{\bf{T}}_{\mathrm{D}}}{{( {n,:} )}^H}} \right).
\end{equation}

On the other hand, the subproblem concerning ${{\bf{T}}_{\mathrm{D}}}$, i.e., \eqref{problem:TDk}, can be rewritten as
\begin{equation}  \label{problem:TD}
\begin{aligned}
    \mathop {\min} \limits_{{{\bf{T}}_{\mathrm{D}}}} &\left\| {{\bf{Z}} - {{\bf{T}}_{\mathrm{RF}}}{{\bf{T}}_{\mathrm{D}}}} \right\|^2 ,
    \;\;\mathrm{s.t.} \ 
    \left\| {{\bf{T}}_{\mathrm{D}}} \right\|^2 = \frac{{N_{\mathrm{RF}}^{\mathrm{t}}}{P_{\mathrm{T}}}}{{M_{\mathrm{T}}}}, 
\end{aligned}
\end{equation}
which is a non-convex 
QCQP. Here, we employ a computationally efficient algorithm to acquire the optimal 
${{{\bf{T}}_{\mathrm{D}}}}$. 

By defining
$\bar{\bf{Q}} ={\bf{T}}_{\mathrm{RF}}^H{{\bf{T}}_{\mathrm{RF}}}$, and
$\bar{\bf{G}} = {\bf{T}}_{\mathrm{RF}}^H{\bf{Z}}$, the objective of \eqref{problem:TD} is reformulated as 
\begin{equation}
    {\mathrm{tr}}({\bf{T}}_{\mathrm{D}}^H \bar{\bf{Q}}{{\bf{T}}_{\mathrm{D}}}) - 2{\cal R}({\mathrm{tr}}({\bf{T}}_{\mathrm{D}}^H \bar{\bf{G}})). 
\end{equation} 
By introducing the Lagrange multiplier $\lambda_L$, the Lagrangian function is formulated as
\begin{equation} \label{eq:Lagfun}
\begin{aligned}
\mathcal{L} ({\bf{T}}_{\mathrm{D}},\lambda_L ) =  & {\mathrm{tr}}({\bf{T}}_{\mathrm{D}}^H \bar{\bf{Q}}{{\bf{T}}_{\mathrm{D}}}) - 2{\cal R}({\mathrm{tr}}({\bf{T}}_{\mathrm{D}}^H{ \bar{\bf{G}}})) \\
&+ \lambda_L (\left\| {{\bf{T}}_{\mathrm{D}}} \right\|^2 -  \frac{{N_{\mathrm{RF}}^{\mathrm{t}}}{P_{\mathrm{T}}}}{{M_{\mathrm{T}}}}) .
\end{aligned}
\end{equation}
According to the optimality condition of \eqref{eq:Lagfun}, we obtain
\begin{equation} \label{eq:TD}
{\bf{T}}_{\mathrm{D}}^{\mathrm{opt}} = { {( { \bar{\bf{Q}} + {\lambda_L^{\mathrm{opt}}}{\bf{I}}} )}^{ - 1}} \bar{\bf{G}},
\end{equation} 
where $\lambda_L ^{\mathrm{opt}}$ can be obtained by bisection search \cite{boyd2004convex}.

In general, the HBF for DFRC systems with PC architectures can be obtained by replacing the solution process of ${{\bf{T}}_{\mathrm{RF}}}$ and ${{\bf{T}}_{\mathrm{D}}}$ in \textbf{Algorithm~\ref{alg1}}, i.e., replacing $\eqref{TRF}$ and $\eqref{eq:TDk}$
using $\eqref{eq:TRF2}$ and $\eqref{eq:TD}$ in Step 3, respectively. The description is omitted because the PC-based receive beamforming can be obtained similarly.

\section{System Performance Analysis}

\subsection{Convergence Analysis}
 
\textbf{Algorithm \ref{alg1}} converges to an effective solution to problem $\mathrm{P}_1$ by iteratively running the proposed double-loop process. Recall that $\mathcal{\tilde{X}} \!=\! \{ { {{\bf{T}}_{\mathrm{RF}}}, {{\bf{T}}_{\mathrm{D},k}},  {\bf{U}}_k, {{\bf{W}}_{\mathrm{RF}}},{{\bf{W}}_{\mathrm{D}}},  {\bf{G}}_k, {{\bf{T}}_k}, {\bf{W}}  } \}$ collects all variables to be updated in the inner-loop, where ${\bf{W}}$, ${\bf{U}}_k$, and ${\bf{G}}_k$ depend uniquely and deterministically on ${\bf{T}}_k$, followed by finding ${{\bf{T}}_{\mathrm{RF}}}{{\bf{T}}_{\mathrm{D},k}}$ and ${{\bf{W}}_{\mathrm{RF}}}{{\bf{W}}_{\mathrm{D}}}$ to approximate ${\bf{T}}_k$ and $\bf{W}$, respectively. In this sense, we can partition $\mathcal{\tilde{X}}$ into three subsets: ${\mathcal{\tilde{X}}}_1 =\{ {\bf{W}}, {\bf{T}}_k, {\bf{U}}_k, {\bf{G}}_k \}$, ${\mathcal{\tilde{X}}}_2 =\{ {{\bf{T}}_{\mathrm{RF}}}, {{\bf{W}}_{\mathrm{RF}}} \}$, and ${\mathcal{\tilde{X}}}_3 = \{ {{\bf{T}}_{\mathrm{D},k}}, {{\bf{W}}_{\mathrm{D}}}\}$, the three of which are updated alternately in each inner-loop iteration.

Given $\rho$, $\bf{D}$, and $\tilde{\bf{D}}$ specified at the beginning of the current outer-loop iteration, ${\mathcal{\tilde{X}}}_1$ depends solely on the SCA process of ${\bf{T}}_k, \forall k$. The objective function \eqref{eq:problem4A} is certainly non-increasing during the process. Given ${\mathcal{\tilde{X}}}_1$, ${\mathcal{\tilde{X}}}_2$ and ${\mathcal{\tilde{X}}}_3$ are updated, which would further decrease \eqref{eq:problem4A} due to the fact that \eqref{TRF} and \eqref{eq:TDk} provide the optimal solutions to \eqref{problem:TRF} and \eqref{problem:TDk}, respectively, and hence ensure the objective function of \eqref{problem:TRF} or \eqref{problem:TDk} is strictly non-increasing. Likewise, \eqref{WRF} and \eqref{eq:WD} can ensure the objective function of \eqref{problem:WRF} is strictly non-increasing. As a result, \eqref{eq:problem4A} is strictly non-increasing during an inner-loop iteration and hence throughout the inner-loop iterations during an outer-loop iteration. On the other hand, \eqref{eq:problem4A} is apparently lower bounded due to the finite and non-negative transmit power of the BS.

In the outer loop, $\rho$, $\bf{D}$, and $\tilde{\bf{D}}$ are updated under the constraint violation condition, ${h(\mathcal{\tilde{X}})}$, until the PDD convergence. While \eqref{eq:problem4A} may not be strictly non-increasing in the outer loop, the outer-loop iterations can increasingly penalize the violation of \eqref{eq:problem3D} and \eqref{eq:problem3E}. When ${\bf{T}}_k$ stops changing in the inner loop, so do the rest of the variables in the inner loop, and the outer-loop iterations also terminate according to \eqref{D1} and \eqref{D2}. With the optimal $\rho^*$, $\bf{D}^*$, and $\tilde{\bf{D}}^*$ in the outer loop, 
\textbf{Algorithm \ref{alg1}} eventually converges.

\subsection{Discussion on Optimality}
Due to the NP-hardness of problem $\mathrm{P}_1$, existing techniques fall short of delivering global optima. Nevertheless, the FD counterpart of the considered system, i.e., problem P$_1$ with \eqref{eq:problem1D}--\eqref{eq:problem1G} dropped, can be solved using \textbf{Algorithm~1} without \eqref{TRF}, \eqref{WRF}, \eqref{eq:TDk} and \eqref{eq:WD} in Step 3, reach a Karush-Kuhn-Tucker (KKT) point, and serve as an upper bound for the proposed HBF algorithm. As evidenced numerically in Section V, the HBF design in \textbf{Algorithm 1} significantly outperforms the existing alternative approaches, much closer to the FD counterpart. In the FD counterpart of the considered problem, only ${\mathcal{\tilde{X}}}_1$ remains to be optimized. As discussed earlier, ${\mathcal{\tilde{X}}}_1$ depends uniquely and deterministically on ${\bf{T}}_k, \forall k$, with the rest of the variables in ${\mathcal{\tilde{X}}}_1$ all given in closed-form functions of ${\bf{T}}_k, \forall k$. To this end, the optimality of ${\mathcal{\tilde{X}}}_1$ depends on that of ${\bf{T}}_k, \forall k$.

Under the FD architecture, the objective function of the ${\bf{T}}_k$ subproblem contains only the first term of \eqref{eq:problemtkA}. The solution yields ${a_{k,i}} = { {{\bf{t}}_i^H} ({\bf{I}} \otimes {\bf{P}}_k^H) {{\bf{t}}_i}  },\, \forall i,k$. Or, one could reduce the auxiliary variable ${a_{k,i}}$ to reduce \eqref{eq:problemtkA}. 
After variable substitution and SCA, the ${\bf{T}}_k$ subproblem is convexified. 
The non-convex constraint $\mathrm{tr} ( {{\bf{T}}_k^H{\bf{\Phi }}{{\bf{T}}_k}} ), \forall k $ is replaced by the first-order Taylor approximation ${2{\cal R} ( {{\bf{\bar t}}_k^H \bar{{\bf{\Phi }}}{{\bf{t}}_k}} ) - {\bf{\bar t}}_k^H \bar{\bf{\Phi}}{{{\bf{\bar t}}}_k}}, \forall k$. Here, $\mathrm{tr} ( {{\bf{T}}_k^H{\bf{\Phi }}{{\bf{T}}_k}} )$ and ${2{\cal R} ( {{\bf{\bar t}}_k^H \bar{{\bf{\Phi }}}{{\bf{t}}_k}} ) - {\bf{\bar t}}_k^H \bar{\bf{\Phi}}{{{\bf{\bar t}}}_k}}$ yield the same value and the same gradient at the local point ${\bf{\bar t}}_k$. 
Moreover, the Slater’s condition is satisfied in every iteration of SCA; i.e., there is a feasible solution to the ${\bf{T}}_k$ subproblem such that all inequality constraints are strictly satisfied. According to~\cite{marks1978general}, this solution is a KKT stationary point, which is a local minimum if it is in the interior of the feasible set of the ${\bf{T}}_k$ subproblem. As a result, the WPDD can achieve a KKT point for the FD counterpart of the considered problem, which is a local minimum if it is inside the feasible region.

\subsection{Computational Complexity}
Since there are typically many more transmit antennas than receive antennas in DFRC systems, the complexity of the WPDD algorithm is predominantly contributed by the solution of ${\bf{T}}_k$ in Step 3. Solving for ${\bf{T}}_k$ in \eqref{problem:SOCP} involves $(K{M_{\mathrm{T}}} + {K^2} + K)$ variables and $({K^2} + K +1)$ SOC constraints. The complexity of running the interior point method to solve for an $\epsilon$-optimal solution for ${\bf{T}}_k$ is $ \mathcal{O} ( { {K^4}{M_{\mathrm{T}}^3}} \log ( 1/ \epsilon ) ),\, \forall k$~\cite{ben2001lectures}. 
The worst-case complexity of WPDD under the RS architecture is $\mathcal{O} \left( { {I_{\mathrm{O}}} {I_{\mathrm{I}}} {\left({ {K^4}{M_{\mathrm{T}}^3} \log ( 1/ \epsilon ) } \right)} }  \right) $, where ${I_{\mathrm{O}}}$ and ${I_{\mathrm{I}}}$ indicate the iteration numbers in the outer and inner loops, respectively.

\subsection{Energy Efficiency (EE)}
\begin{table}[tbp] 
    \centering
    \caption{The total power consumption under different beamforming architectures}  \label{tab:EE}
    \renewcommand\arraystretch{1.5}
    \resizebox{\linewidth}{!}{
    \begin{tabular}{c|c}
        \hline
        \textbf{Architecture} & \textbf{Total Power Consumption}\\   
         \hline
         RS & $ P_\mathrm{tot}^{\mathrm{RS}} = {P_{\mathrm{T}}} \!+ \!{P_{\mathrm{BB}}} \!+\! ({N_{\mathrm{RF}}^{\mathrm{t}}}\! + \!{N_{\mathrm{RF}}^{\mathrm{r}}} ){P_{\mathrm{RF}}} \!+ \!({M_{\mathrm{T}}} \!+\! {M_{\mathrm{R}}}){P_{\mathrm{PS}}} \!+\! ({M_{\mathrm{T}}} \!+ \!{M_{\mathrm{R}}}){P_{\mathrm{SW}}}$ \\   
         PC & $ P_\mathrm{tot}^{\mathrm{PC}} = {P_{\mathrm{T}}} + {P_{\mathrm{BB}}} + ({N_{\mathrm{RF}}^{\mathrm{t}}}\! + \!{N_{\mathrm{RF}}^{\mathrm{r}}} ){P_{\mathrm{RF}}} + ({M_{\mathrm{T}}} \!+\! {M_{\mathrm{R}}}){P_{\mathrm{PS}}} $  \\
         DPC & $ P_\mathrm{tot}^{\mathrm{DPC}} = {P_{\mathrm{T}}} + {P_{\mathrm{BB}}} + ({N_{\mathrm{RF}}^{\mathrm{t}}}\! + \!{N_{\mathrm{RF}}^{\mathrm{r}}} ){P_{\mathrm{RF}}} + 2({M_{\mathrm{T}}} \!+\! {M_{\mathrm{R}}}){P_{\mathrm{PS}}} $  \\
         FC  & $ P_\mathrm{tot}^{\mathrm{FC}} = {P_{\mathrm{T}}} + {P_\mathrm{BB}} + ({N_{\mathrm{RF}}^{\mathrm{t}}}\! + \!{N_{\mathrm{RF}}^{\mathrm{r}}} ){P_{\mathrm{RF}}} + ({M_{\mathrm{T}}}{N_{\mathrm{RF}}^{\mathrm{t}}} +{M_{\mathrm{R}}}{N_{\mathrm{RF}}^{\mathrm{r}}} ){P_{\mathrm{PS}}} $  \\
         FD & $ P_\mathrm{tot}^{\mathrm{FD}} = {P_{\mathrm{T}}} + {P_\mathrm{BB}} + ({M_{\mathrm{T}}} \!+\! {M_{\mathrm{R}}}){P_{\mathrm{RF}}} $  \\
         \hline         
    \end{tabular}   }
\end{table}

To evaluate the effect of the RS HAD DFRC system from the perspectives of hardware costs and power consumption, the average EE of each user is defined as 
\begin{equation}
    \eta \buildrel \Delta \over = \frac{ 1 }{K }\frac{ {R} }{P_\mathrm{tot}},
\end{equation}
where $P_\mathrm{tot}$ is the total energy consumption of the DFRC BS. 

The overall power consumption under various beamforming architectures with transmit power $P_\mathrm{T}$ is summarized in Table~\ref{tab:EE}, where $P_{\mathrm{BB}}$ specifies the circuit power consumed by the baseband circuit; $P_{\mathrm{RF}}$ and $P_{\mathrm{SW}}$ are the powers consumed by an RF chain and a switch, respectively; $P_{\mathrm{PS}}$ gives the power consumed by a PS; $N_{\mathrm{PS}}$ is the number of PSs.
Clearly, the FD has the maximum power consumption because there are the same number of RF chains and transmit/receive antennas. Moreover, the number of PSs is proportional to that of the transmit/receive antennas and depends on the specific connection mode in the HAD architecture.

\section{Numerical Results}

In this section, the developed algorithm for DFRC systems with HBF architectures is numerically assessed by MATLAB simulations. The path loss $PL(d_k)$ is modeled as 
\begin{equation}
    PL(d_k) \ [\text{dB}] = \alpha + 10 \beta \log_{10}(d_k) + \xi,
\end{equation}
where $\xi \sim {\cal {N}} (0,\sigma^2)$, $\alpha = 72.0$, $\beta = 2.92$, and $\sigma = 8.7$ dB for a non-line-of-sight (NLoS) path \cite{6834753}. The noise power of user $k$ is $\sigma _k^2 = -90$ dBm. The AoDs $\phi_{k,l},\,\forall k,l$, follow a uniform distribution within $[- \frac{\pi}{ 2}, \frac{\pi}{ 2} ]$. The other simulation parameters are collated in Table \ref{tab:SimPara}. 
For the WPDD algorithm, the initial penalty factor is $\rho=1$, the control parameter is $c=0.6$, and the termination tolerances are $\varepsilon _1 = \varepsilon _2 = 10^{-4}$, where $\varepsilon _1$ controls the accuracy of inner iterations. 
Moreover, we set the maximum number of iterations $I_{\mathrm{max}}$ = 30 to prevent slow convergence in the inner loop. All of the simulations are conducted using the Matlab 2018b version using a standard PC with an i7-8700 3.20 GHz CPU, 16 GB RAM, and a 64-bit operating system.

\begin{table}[t]
	\caption{Simulation parameters} \label{tab:SimPara}
	\centering
    \scalebox{0.9}{
    \begin{tabular}{l|l||l|l}
		\hline 
		\textbf{Parameter} & \textbf{Value} & \textbf{Parameter} & \textbf{Value}  \\
		\hline    
        $M_{\mathrm{T}}$  & {64} & $M_{\mathrm{R}}$  & 16 \\
        ${N_{\mathrm{RF}}^{\mathrm{t}}}$  & {8} & ${N_{\mathrm{RF}}^{\mathrm{r}}}$  & {8} \\
        $M_{\mathrm{U}}$ & 2 & $P_{\mathrm{T}}$  & 40 dBm \\ 
        $K$  & 4 &  $d_s$  & 2   \\
        $L_k$  & 3  &  $d_k$  & 80 m  \\
        $\theta _0$  & $0^\circ$ & 
        $\theta _j$  & $\{-30^\circ,30^\circ\}$  \\
        $\sigma _{\mathrm{0}}^2$  & 10 dB & $\sigma _{\mathrm{C}}^2$  & 20 dB  \\       
		\hline
	\end{tabular} }
\end{table}

The schemes used for evaluating and benchmarking the system performance are outlined as follows:
\begin{itemize}
 \item{RS-WPDD}: This is the proposed WPDD algorithm for HBF design in multi-user DFRC systems with the RS architectures (see Section III-C). 
 \item{PC-WPDD}: This is the adaptation of the proposed WPDD algorithm for multi-user DFRC systems with the PC architectures (see Section III-D). 
 \item{DPC-THEREON}: This is the joint hybrid waveform and radar receiver optimization algorithm in \cite{9950549} based on the DPS architectures.
 \item{FC-TwoStage}: This is the FC-based HBF algorithm developed in \cite{9729809}, which first constructs an FD beamformer $\mathbf{T}$ and then iteratively optimizes the analog and digital beamformers to approximate the achieved $\mathbf{T}$. 
 \item{PC-TwoStage}: This is the HBF algorithm developed under the PC architectures in \cite{8683591}, which can be implemented similarly to FC-TwoStage algorithm.
 \item{FD/FA}: This is the fully-digital/analog-based beamforming algorithm that provides an upper/lower bound for all HBF algorithms.

\end{itemize}


\begin{figure}[t]
	\centering
    {\includegraphics[scale=0.6]{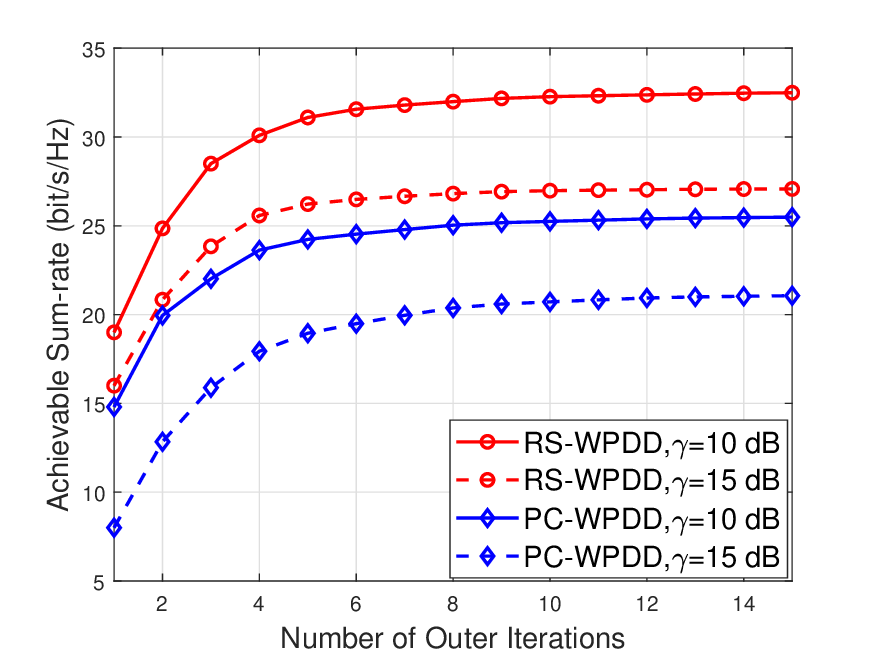}}  \captionsetup{justification=raggedright,font={small}}
	\caption{Convergence behaviors of the WPDD algorithm with $\gamma$ = 10 dB, and 15 dB.}    
\label{fig_Convergence}
\end{figure}

Fig. \ref{fig_Convergence} plots the convergence of the proposed WPDD algorithm for the HBF design with the increase of outer iterations, where the radar SCNR threshold~$\gamma$ is set to 10 dB and 15 dB. It is noticed that the communication sum-rate converges within about ten iterations under both the RS-WPDD and PC-WPDD algorithms. 
Moreover, RS-WPDD always obtains better communication performance than PC-WPDD.

\begin{figure}[t]
	\centering
	\includegraphics[scale=0.6]{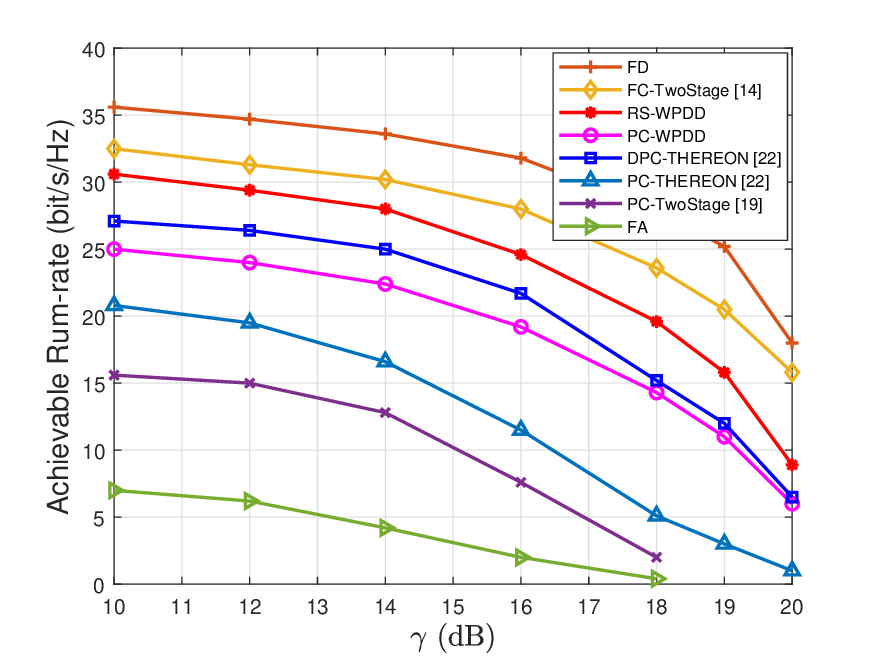} \captionsetup{justification=raggedright,font={small}}
	\caption{Trade-off between the communication sum-rate and the radar receive SCNR $\gamma$, under different beamforming design schemes. $M_{\mathrm{T}}$ = 64, $M_{\mathrm{R}}$ = 16,  ${N_{\mathrm{RF}}^{\mathrm{t}}}={N_{\mathrm{RF}}^{\mathrm{r}}}=8$, $K=4$ and $M_{\mathrm{U}}= d_s=2$.}
	\label{fig_SR_Rsnr}
\end{figure}

\begin{figure}
	\centering
	\includegraphics[scale=0.6]{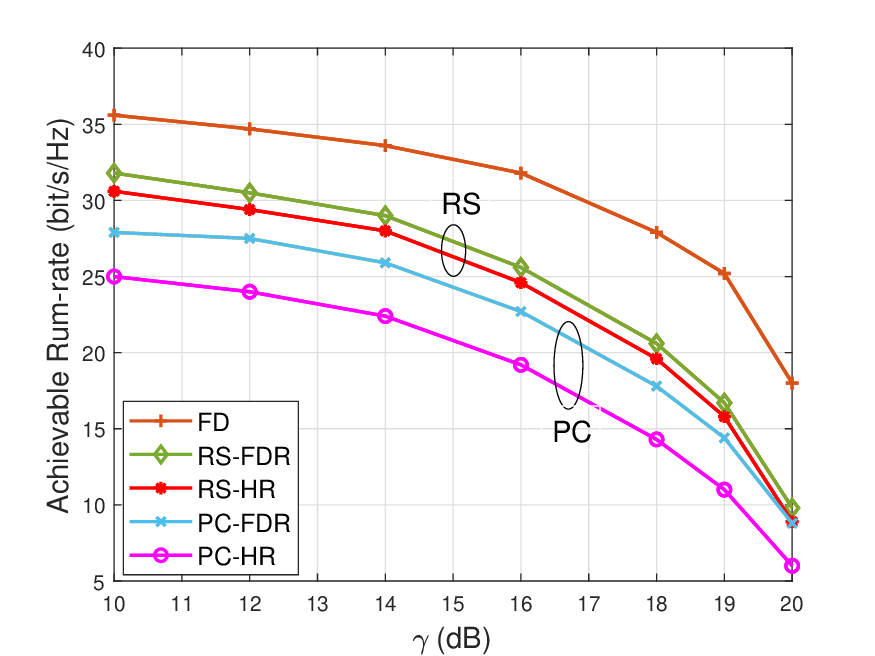} \captionsetup{justification=raggedright,font={small}}
	\caption{ Trade-off between the communication sum-rate and the radar receive SCNR $\gamma$, under hybrid/FD DFRC receivers. $M_{\mathrm{T}}$ = 64, $M_{\mathrm{R}}$ = 16,  ${N_{\mathrm{RF}}^{\mathrm{t}}}={N_{\mathrm{RF}}^{\mathrm{r}}}$ = 8, $K=4$ and $M_{\mathrm{U}}= d_s=2$.}
	\label{fig_SR_FDR}
\end{figure} 

\begin{figure}[tbp]
	\centering
    {\includegraphics[scale=0.6]{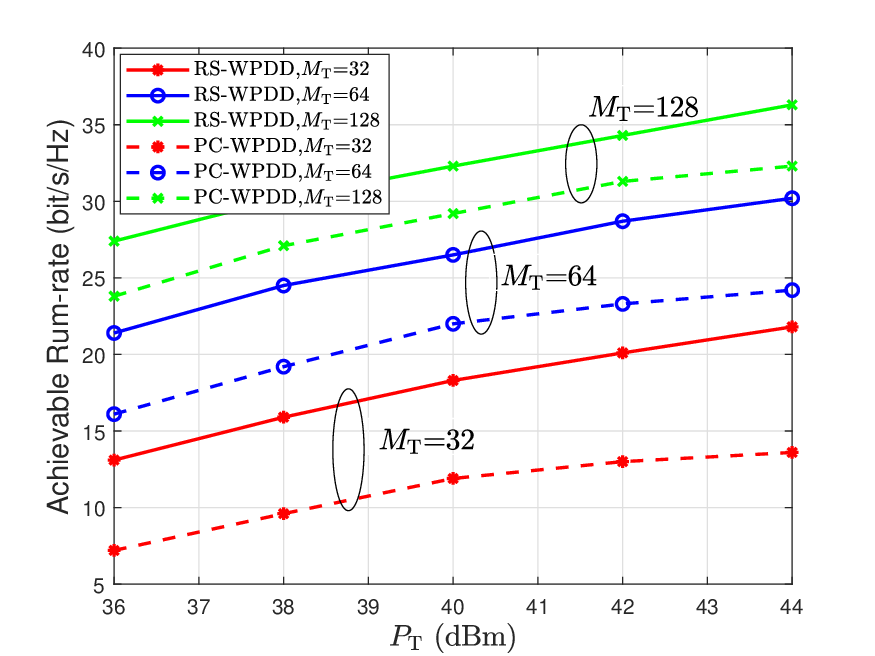}}   \captionsetup{justification=raggedright,font={small}}
	\caption{Achievable sum-rate of communication versus different transmit powers $P_{\mathrm{T}}$, with $\gamma$ = 15 dB, $M_{\mathrm{T}}$ = 32, 64, and 128.}
    \label{fig_SR_Pt}
\end{figure}

Fig. \ref{fig_SR_Rsnr} plots the trade-off between the achievable sum-rate and the radar SCNR constraint under different beamforming design schemes, where $\gamma$ ranges from 10 dB to 20 dB. The sum-rate declines as~$\gamma$ grows. The reason is that when the sensing demand is higher, fewer DoFs are used to improve communication performance. When $\gamma\leq 12$ dB, the multi-user communication sum-rate remains approximately constant since the system is saturated when the sensing performance is easily satisfied with small $\gamma$. When $\gamma \geq 16$ dB, the communication performance degrades rapidly. In addition, Fig. \ref{fig_SR_Rsnr} shows that the proposed RS-WPDD algorithm outperforms the PC-based algorithms. Even the proposed PC-WPDD has comparable performance with DPC-THEREON. Here, FD offers an upper bound for the performance of our proposed system, while FA gives a lower bound.

Fig. \ref{fig_SR_FDR} shows the trade-off between the communication and sensing of the WPDD framework under different DFRC radar receive architectures, including the hybrid receiver (HR) architecture and the fully digital receiver (FDR) architecture. The RS-based DFRC hybrid receive beamforming is closer to the FD receive beamforming than the PC architecture. Fig. \ref{fig_SR_Pt} plots the achievable sum-rate of the DFRC system with the increasing total transmit power $P_{\mathrm{T}}$, where the number of transmit antennas is $M_{\mathrm{T}}=32,\,64$, and 128. It is shown that the proposed RS-WPDD is better than PD-WPDD, as observed in Fig. \ref{fig_SR_Rsnr}. The system performance improves as $M_{\mathrm{T}}$ increases, enhancing antenna diversity and better beamforming gain. Furthermore, the difference between RS-WPDD and PC-WPDD decreases as $M_{\mathrm{T}}$ increases.


\begin{figure}
    {\includegraphics[scale=0.6]{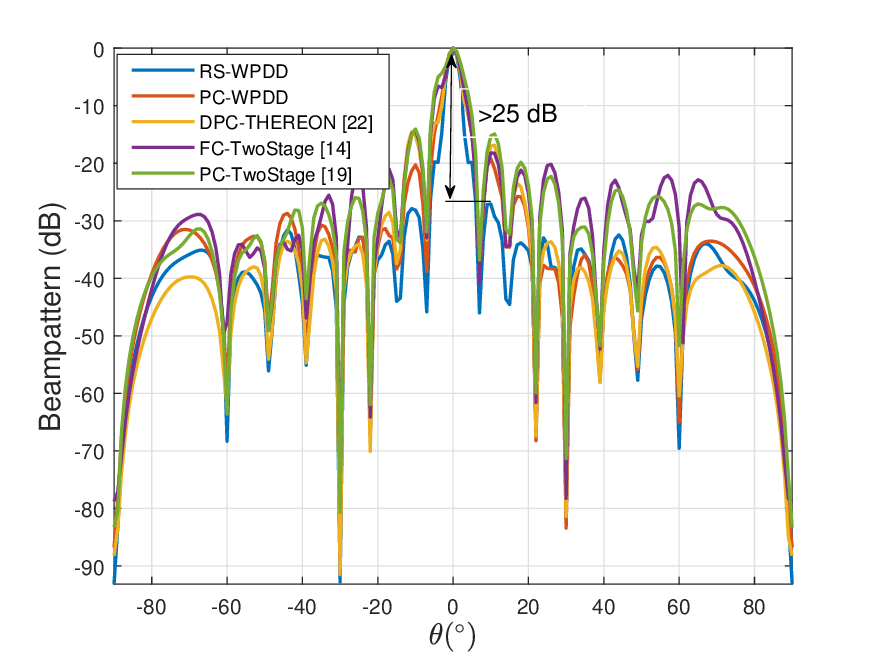}}    \captionsetup{justification=raggedright,font={small}}
	\caption{Beampattern for different architectures with $M_{\mathrm{T}}$ = 32, $M_{\mathrm{R}}$ = 16, $\gamma$ = 18 dB, ${\theta}_0 = 0^\circ$, and ${\theta}_j = \{-30^\circ$, $30^\circ\}$. } 
    \label{fig_BP_Com}
\end{figure}

\begin{figure}
	\centering
	\includegraphics[scale=0.6]{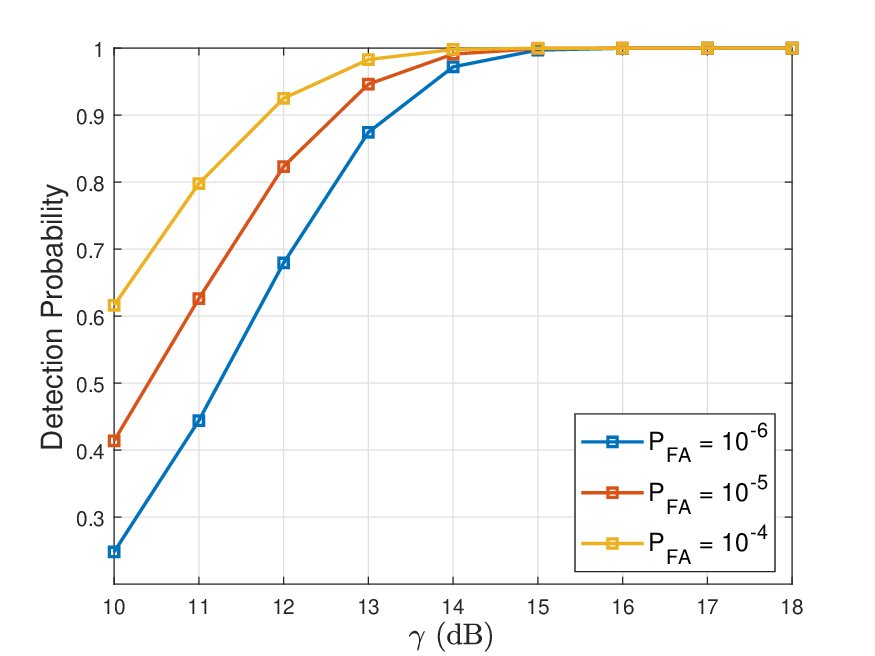}    \captionsetup{justification=raggedright,font={small}}
	\caption{ Detection probability with $M_{\mathrm{T}}$ = 64, $M_{\mathrm{R}}$ = 16, ${\theta}_0 = 0^\circ$, and ${\theta}_j = \{-30^\circ$, $30^\circ\}$. } 
	\label{fig_PD}
\end{figure}

We adopt the spatial beampattern, $P( \theta ) = {\left| {{{\bf{w}}^H}{\bf{A}} ( \theta ){\bf{x}}} \right|^2}$ \cite{9537599}, as a sensing performance metric. The optimized beampatterns of the proposed WPDD algorithm and the benchmarks with $M_{\mathrm{T}} = 32$, $M_{\mathrm{R}} = 16$, and $\gamma = 18$ dB are depicted in Fig.~\ref{fig_BP_Com}. Clearly, all schemes can be pointed at a $0^\circ$ target and achieve good suppression of clutters at $-30^\circ$ and $30^\circ$. Clutters can be effectively suppressed through receive beamforming at the sensing receiver. Moreover, the RS-WPDD scheme has a peak-to-side-lobe ratio of about 30 dB, which is much better than the rest of the schemes.

We also plot the detection probability of the proposed RS-WPDD with the SCNR threshold $\gamma$, as shown in Fig. \ref{fig_PD}. When $\gamma$ = 15 dB, the detection probability is 99.72\% under a false alarm probability of $10^{-6}$, and 99.99\% under a false alarm probability of $10^{-4}$. As shown in Fig. \ref{fig_SR_Rsnr}, RS-WPDD can obtain a satisfactory sensing performance with marginal or even negligible communication rate losses when $\gamma$ = 15 dB.


\begin{figure}
	\centering

    {\includegraphics[scale=0.6]{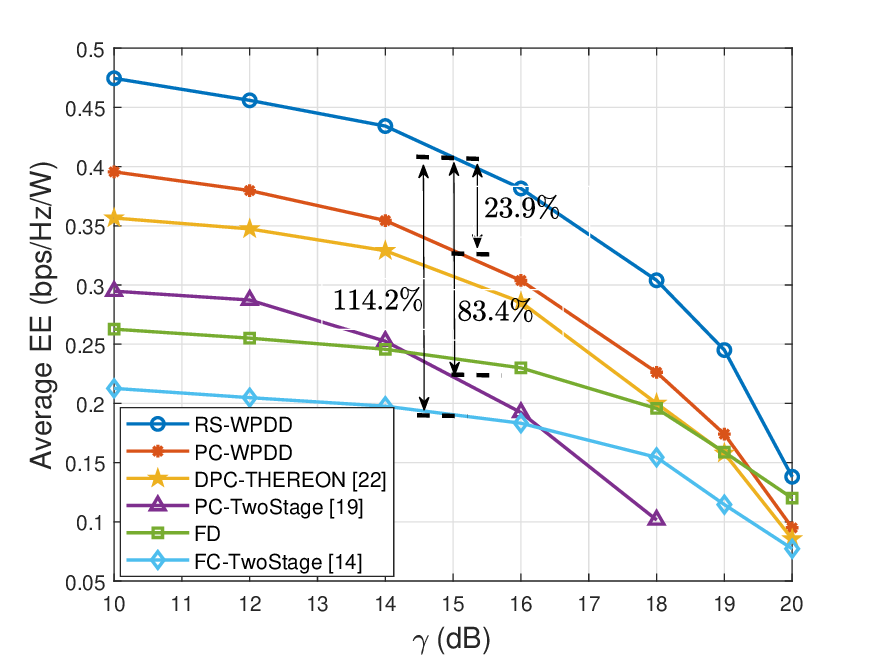}} \captionsetup{justification=raggedright,font={small}}    
	\caption{Average EE vs. the threshold of radar receive SCNR $\gamma$, with $K$ = 4.}
    \label{fig_EE_Rsnr}
\end{figure}

\begin{figure}
	\centering
    {\includegraphics[scale=0.6]{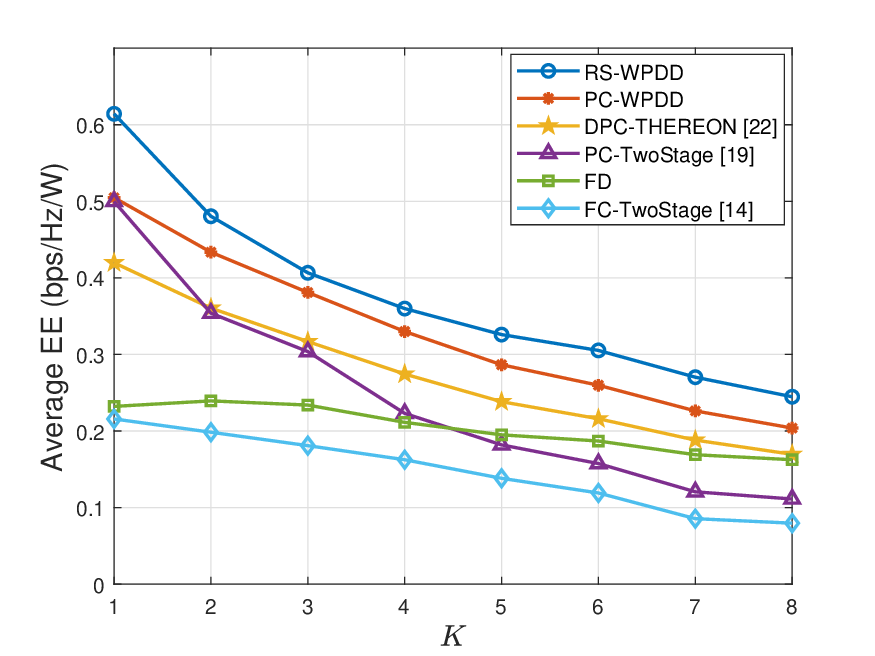}}
    \captionsetup{justification=raggedright,font={small}}
	\caption{Average EE vs. the number of users $K$, with $\gamma$ = 10 dB and $d_s$ = 1.}
    \label{fig_EE_K}
\end{figure}

We evaluate EE to show the effect of our proposed HBF on the system cost and performance, where $P_{\mathrm{BB}} = 200$ mW, $P_{\mathrm{RF}} = 300$ mW, $P_{\mathrm{PS}} = 50$ mW, and $P_{\mathrm{SW}} = 5$ mW~\cite{9110865}.
Fig.~\ref{fig_EE_Rsnr} plots the average EE of the DRFC systems with the increase of $\gamma$. We see that the EE of the proposed RS-WPDD is better than that of PC-WPDD and significantly surpasses the other benchmarks. Specifically, when $\gamma=15$ dB, RS-WPDD improves the EE by 23.9\% compared to PC-WPDD, 83.4\% compared to PC-TwoStage, and 114.2\% compared to FC-TwoStage. Fig. \ref{fig_EE_K} shows the change in EE with the number of users $K \, (K \le {N_{\mathrm{RF}}^{\mathrm{t}}})$, when each user has only one data stream, i.e., $d_s=1$, and ${N_{\mathrm{RF}}^{\mathrm{t}}} = 8$. The RS-WPDD still yields the best EE across the spectrum of the user number, as observed in Fig. \ref{fig_EE_Rsnr}. 
The system parameters can be selected to balance the hardware cost and system performance of DFRC with the RS architectures.

\section{Conclusion}
We have designed HBF for a multi-user mmWave DFRC system with a new RS HAD architecture. 
We have jointly optimized digital and analog beamformers to maximize the communication sum-rate under a prescribed SCNR of target detection. An efficient WPDD algorithm has been designed based on WMMSE and PDD methods, which can guarantee convergence to a stationary point. The effectiveness of WPDD has been confirmed in terms of convergence, communication SE, target beampattern, and system EE. Simulations have shown that with the RS architecture and HBF, DFRC systems excel in both SE and EE compared to the conventional PC architecture. The EE of the systems is dramatically improved compared to the FC architecture by 114.2\%.

\appendices

\section{Proof Of Proposition \ref{Proposition2}} 
\label{appendix_P2}
Because the reconfigurable subarrays do not overlap and every transmit antenna is linked to only an RF chain, we can solve problem (\ref{problem:TRF}) row by row, as follows:
\begin{equation} \label{DSmap1}
\begin{aligned}
&\mathop{\min} \limits_{{\bf{T}}_{\mathrm{RF}}} \left\| {{\bf{Z}} - {{\bf{T}}_{\mathrm{RF}}}{{\bf{T}}_{\mathrm{D}}}} \right\|^2\\
\mathop  = \limits^ {(a)} & \mathop{\min} \limits_{{\bf{T}}_{\mathrm{RF}}} \sum\limits_{n=1}^{N_{\mathrm{RF}}^{\mathrm{t}}} {\sum\limits_{{\forall m} \in {{\cal S}_n}} {\left\| {{\bf{Z}}\left( {m,:} \right) - {{\bf{T}}_{\mathrm{RF}}}\left( {m,n} \right){{\bf{T}}_{\mathrm{D}}}\left( {n,:} \right)} \right\|^2} } \\
=& \mathop{\min}\limits_{{{\bf{T}}_{\mathrm{RF}}}} \sum\limits_{n = 1}^{N_{\mathrm{RF}}^{\mathrm{t}}} \sum\limits_{{\forall m} \in {{\cal S}_n}} \{ \left\| {{\bf{Z}}\left( {m,:} \right)} \right\|^2  + {{{\left| {{{\bf{T}}_{\mathrm{RF}}}\left( {m,n} \right)} \right|}^2}\left\| {{{\bf{T}}_{\mathrm{D}}}\left( {n,:} \right)} \right\|^2 \}}  \\
&- 2{\cal R}\left( {{\bf{Z}}\left( {m,:} \right){{\bf{T}}_{\mathrm{D}}}{{\left( {n,:} \right)}^H}{{\bf{T}}_{\mathrm{RF}}}{{\left( {m,n} \right)}^H}} \right),
\end{aligned}
\end{equation}
where the equality $(a)$ holds when the objective function is expanded in rows. \par

Define ${\varphi _{m,n}}$ as the phase of the ($m,n$)-th entry in ${{\bf{T}}_{\mathrm{RF}}}$, i.e., $ {{\bf{T}}_{\mathrm{RF}}}\left( {m,n} \right) = \frac{1}{{\sqrt {M_{\mathrm{T}}} }}{e^{j{\varphi _{m,n}}}}$, and ${\bf{Z}}\left( {m,:} \right){{\bf{T}}_{\mathrm{D}}}{\left( {n,:} \right)^H} = \left| {{\zeta _{m,n}}} \right|{e^{j{\phi _{m,n}}}}, n =  1,\cdots,{N_{\mathrm{RF}}^{\mathrm{t}}} , m \in {{\cal S}_n}$. By suppressing the constant term, \eqref{DSmap1} is further equivalent to
\begin{equation}\label{DSmap2}
\begin{aligned}
&\mathop{\max} \limits_{{\bf{T}}_{\mathrm{RF}}} \sum\limits_{n =1}^{N_{\mathrm{RF}}^{\mathrm{t}}} {\sum\limits_{{\forall m} \in {{\cal S}_n}} {{\cal R}\left( {{\bf{Z}}\left( {m,:} \right){{\bf{T}}_{\mathrm{D}}}{{\left( {n,:} \right)}^H} {{\bf{T}}_{\mathrm{RF}}}{{\left( {m,n} \right)}^H}} \right)} } \\
\mathop  = \limits^ {(b)}& \mathop{\max} \limits_{\left\{ {\varphi _{m,n}} \right\}} \sum\limits_{n =1}^{N_{\mathrm{RF}}^{\mathrm{t}}} {\sum\limits_{{\forall m} \in {{\cal S}_n}} {{\cal R}\left( {{e^{ - j{\varphi _{m,n}}}}{\bf{Z}}\left( {m,:} \right){{\bf{T}}_{\mathrm{D}}}{{\left( {n,:} \right)}^H}} \right)} } \\
=& \mathop{\max} \limits_{\left\{ {\varphi _{m,n}} \right\}} \sum\limits_{n =1}^{N_{\mathrm{RF}}^{\mathrm{t}}} {\sum\limits_{{\forall m} \in {{\cal S}_n}} {\left| {\zeta _{m,n}} \right| \cos \left( {{\phi _{m,n}} - {\varphi _{m,n}}} \right)} }.
\end{aligned}
\end{equation}
Since the ($m,n$)-th entry of ${{\bf{T}}_{\mathrm{RF}}}$ has a constant modulus value, the equality in (b) holds. The final establishment of the equivalence between problems \eqref{problem:TRF} and \eqref{problem:DSmap} is achieved by integrating the objective \eqref{DSmap2} with the reconfigurable connection constraints \eqref{eq:DCcons1} and \eqref{eq:DCcons2}.

\ifCLASSOPTIONcaptionsoff
 \newpage
\fi



\begin{IEEEbiography}[{\includegraphics[width=1in,height=1.25in,clip,keepaspectratio]{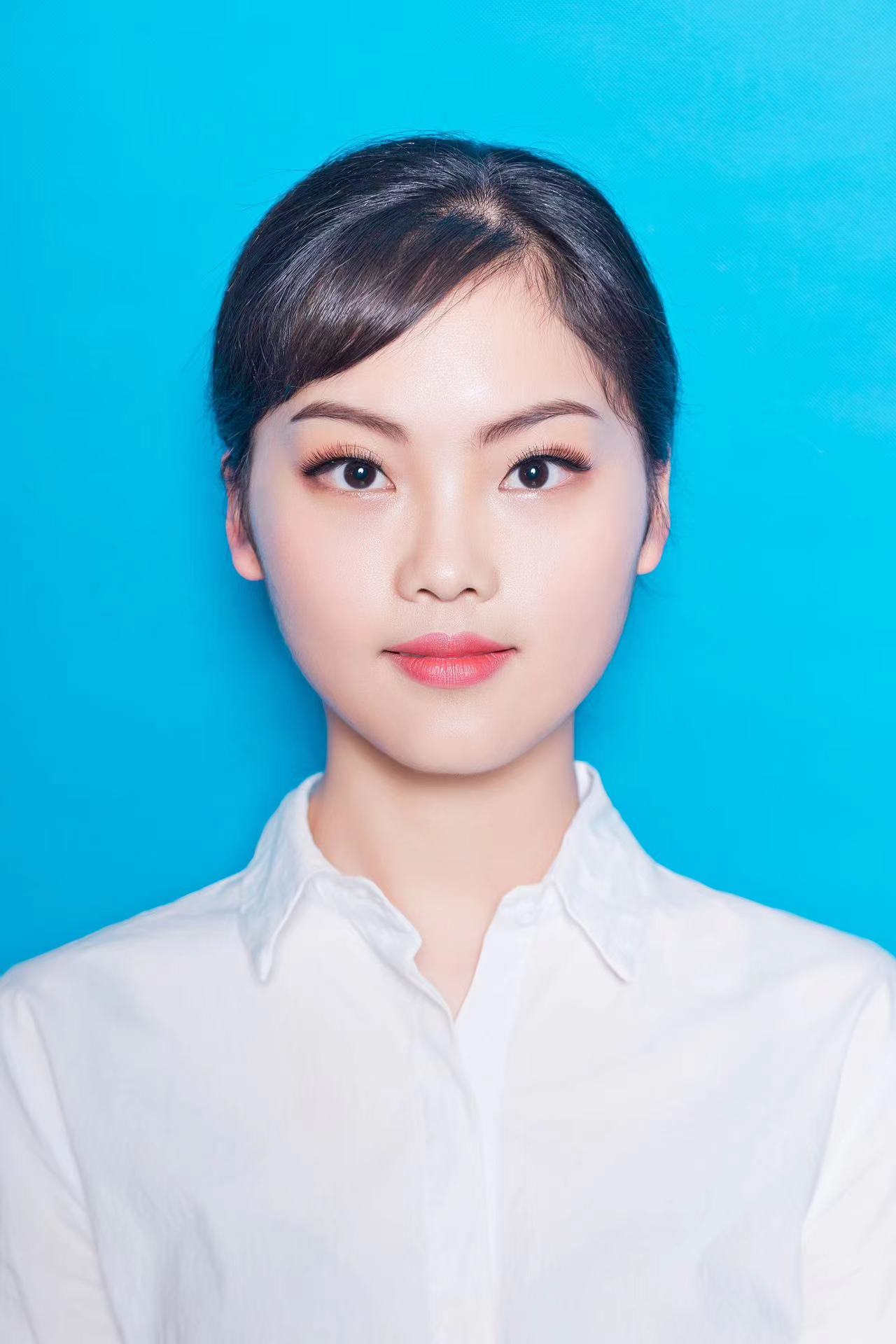}}]{Xin Jin} received the B.S. degree in communication engineering from Sichuan Normal University (SICNU), Chengdu, China, in 2020. She is currently pursuing the Ph.D. degree with the School of Information and Communication Engineering, Beijing University of Posts and Telecommunications (BUPT), Beijing, China. Her research interests include integrated sensing and communications, massive MIMO, hybrid beamforming, and array signal processing.

\end{IEEEbiography}


\begin{IEEEbiography}[{\includegraphics[width=1in,height=1.25in,clip,keepaspectratio]{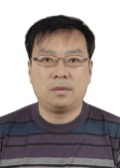}}]{Tiejun Lv}(Senior Member, IEEE) received the M.S. and Ph.D. degrees in electronic engineering from the University of Electronic Science and Technology of China (UESTC), Chengdu, China, in 1997 and 2000, respectively. From January 2001 to January 2003, he was a Postdoctoral Fellow at Tsinghua University, Beijing, China. In 2005, he was promoted to Full Professor at the School of Information and Communication Engineering, Beijing University of Posts and Telecommunications (BUPT). From September 2008 to March 2009, he was a Visiting Professor with the Department of Electrical Engineering at Stanford University, Stanford, CA, USA. He is the author of four books, more than 150 published journal papers and 200 conference papers on the physical layer of wireless mobile communications. His current research interests include signal processing, communications theory and networking. He was the recipient of the Program for New Century Excellent Talents in University Award from the Ministry of Education, China, in 2006. He received the Nature Science Award from the Ministry of Education of China for the hierarchical cooperative communication theory and technologies in 2015. He has won best paper award at CSPS 2022. 
\end{IEEEbiography}


\begin{IEEEbiography}[{\includegraphics[width=1in,height=1.25in,clip,keepaspectratio]{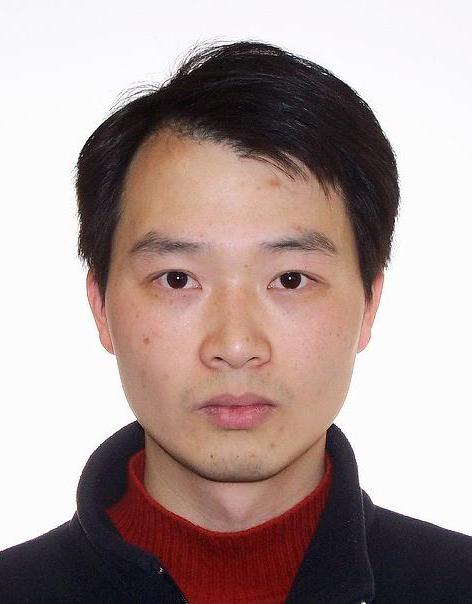}}]{Wei Ni} (Fellow, IEEE) received the B.E. and Ph.D. degrees in Electronic Engineering from Fudan University, Shanghai, China, in 2000 and 2005, respectively. He is a Principal Research Scientist at CSIRO, Sydney, Australia, and a Conjoint Professor at the University of New South Wales. He is also an Adjunct
Professor at the University of Technology Sydney, and an Honorary Professor at Macquarie University. He serves as a Technical Expert at Standards Australia in support of the ISO standardization of AI and Big Data. He was a Postdoctoral Research Fellow at Shanghai Jiaotong University from 2005 to 2008; Deputy Project Manager at Bell Labs, Alcatel/Alcatel-Lucent from 2005 to 2008; and Senior Researcher at Devices R\&D, Nokia from 2008 to 2009. He has co-authored one book, ten book chapters, more than 300 journal papers, more than 100 conference papers, 26 patents, ten standard proposals accepted by IEEE, and three technical contributions accepted by ISO. His research interests include 6G security and privacy, machine learning, stochastic optimization, and their applications to system efficiency and integrity.

Dr. Ni has been an Editor for IEEE Transactions on Wireless Communications since 2018, an Editor for IEEE Transactions on Vehicular Technology since 2022, and an Editor for IEEE Transactions on Information Forensics
and Security and IEEE Communications Surveys and Tutorials since 2024. He served first as the Secretary, then the Vice-Chair and Chair of the IEEE VTS NSW Chapter from 2015 to 2022, Track Chair for VTC-Spring 2017, Track Co-chair for IEEE VTC-Spring 2016, Publication Chair for BodyNet 2015, and Student Travel Grant Chair for WPMC 2014.
\end{IEEEbiography}


\begin{IEEEbiography}[{\includegraphics[width=1in,height=1.25in,clip,keepaspectratio]{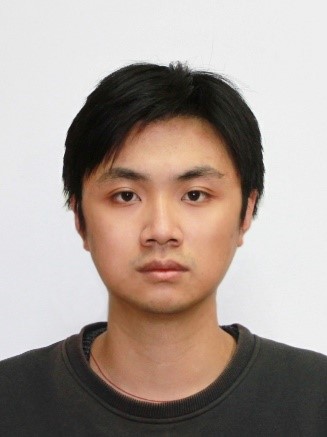}}]{Zhipeng Lin}(Member, IEEE) received the Ph.D. degrees from the School of Information and Communication Engineering, Beijing University of Posts and Telecommunications, Beijing, China, and the School of Electrical and Data Engineering, University of Technology of Sydney, NSW, Australia, in 2021. Currently, He is an Associate Researcher in the College of Electronic and Information Engineering, Nanjing University of Aeronautics and Astronautics, Nanjing, China. His current research interests include signal processing, massive MIMO, spectrum sensing, and UAV communications.
\end{IEEEbiography}


\begin{IEEEbiography}[{\includegraphics[width=1in,height=1.25in,clip,keepaspectratio]{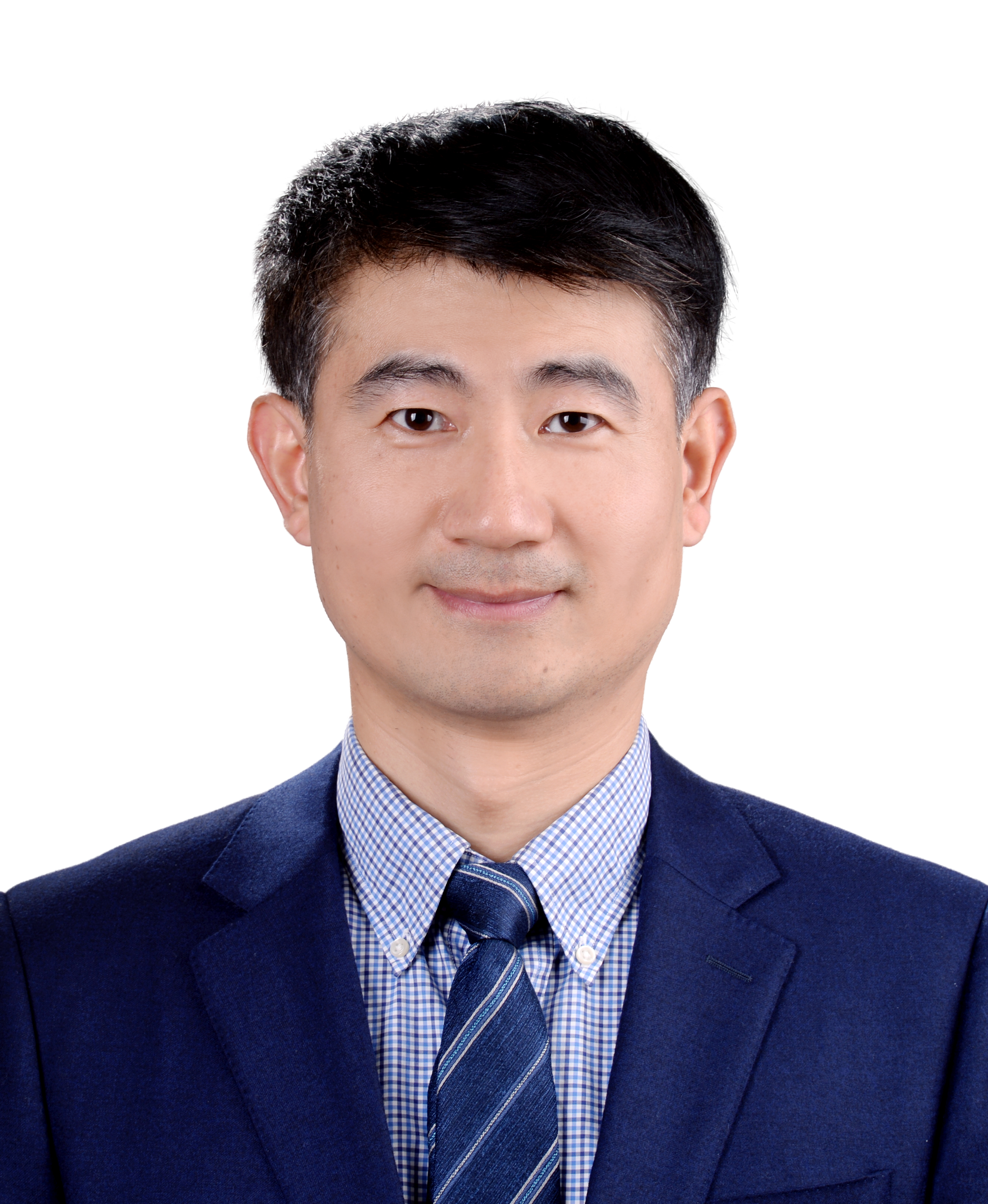}}]{Qiuming Zhu}(Senior Member, IEEE) received his B.S. in electronic engineering from Nanjing University of Aeronautics and Astronautics (NUAA), Nanjing, China, in 2002 and his M.S. and Ph.D. in communication and information system from NUAA in 2005 and 2012, respectively. Since 2021, he has been a professor in the Department of Electronic Information Engineering, NUAA. From Oct. 2016 to Oct. 2017, June 2018 to Aug. 2018 and June 2018 to Aug. 2018, he was also an academic visitor at Heriot Watt University,
Edinburgh, U. K. He has authored or coauthored more than 120 articles in refereed journals and conference proceedings and holds over 40 patents. His current research interests include channel sounding, modeling, and emulation for the fifth/sixth generation (5G/6G) mobile communication, vehicle-to-vehicle (V2V) communication and unmanned aerial vehicles (UAV) communication systems.
\end{IEEEbiography}


\begin{IEEEbiography}[{\includegraphics[width=1in,height=1.25in,clip,keepaspectratio]{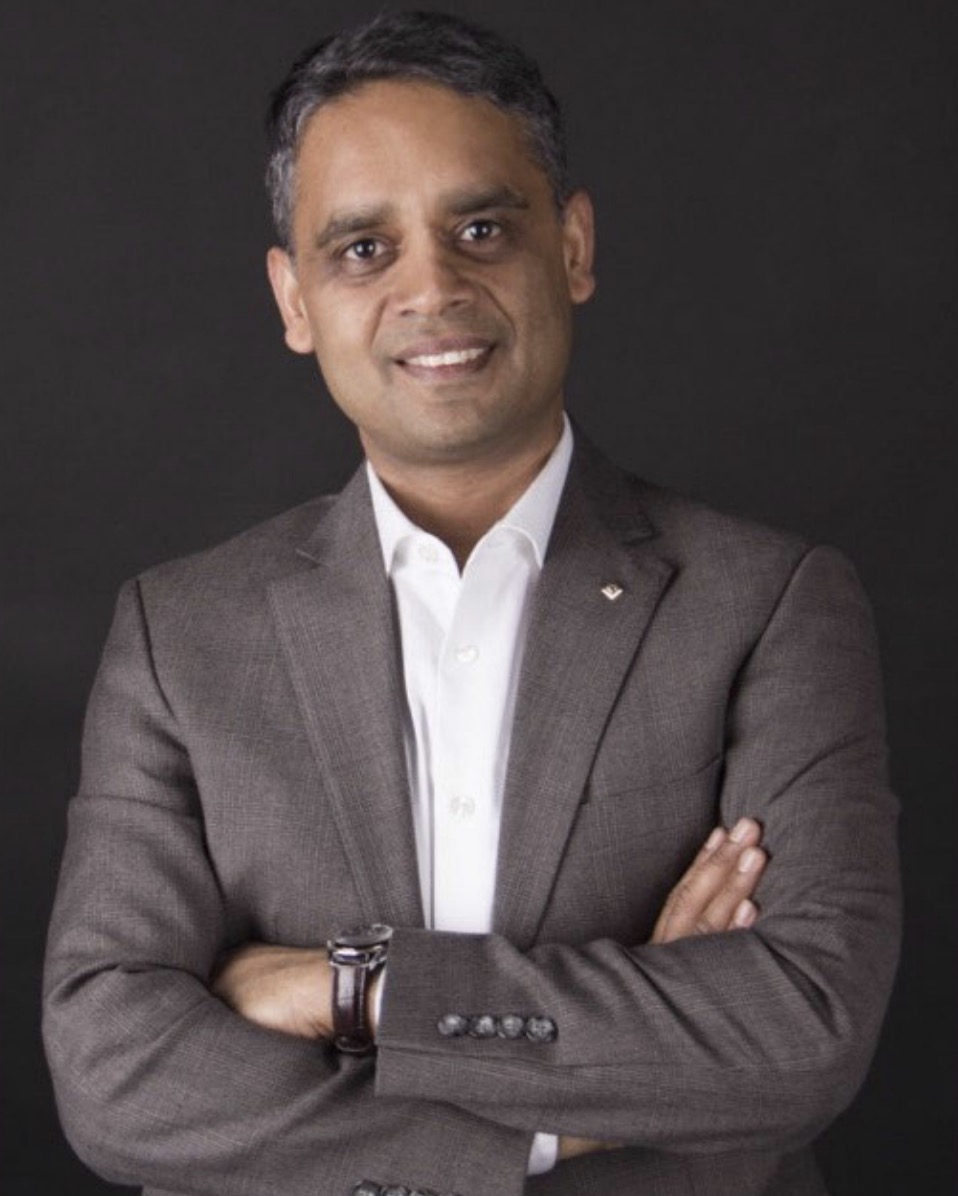}}]{Ekram Hossain}(Fellow, IEEE) is a Professor and the Associate Head (Graduate Studies) of the Department of Electrical and Computer Engineering, University of Manitoba, Canada. He is a Member (Class of 2016) of the College of the Royal Society of Canada. He is also a Fellow of the Canadian Academy of Engineering and the Engineering Institute of Canada. He has won several research awards, including the 2017 IEEE Communications Society Best Survey Paper Award and the 2011 IEEE Communications Society Fred Ellersick Prize Paper Award. He was listed as a Clarivate Analytics Highly Cited Researcher in Computer Science in 2017-2023. Previously, he served as the Editor-in-Chief (EiC) for the IEEE Press (2018–2021) and the IEEE Communications Surveys and Tutorials (2012–2016). He was a Distinguished Lecturer of the IEEE Communications Society and the IEEE Vehicular Technology Society. He served as the Director of Magazines (2020-2021) and the Director of Online Content (2022-2023) for the IEEE Communications Society.
\end{IEEEbiography}


\begin{IEEEbiography}[{\includegraphics[width=1in,height=1.25in,clip,keepaspectratio]{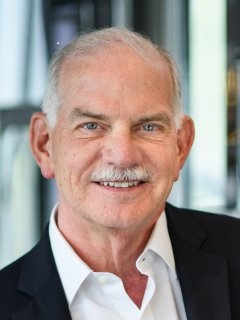}}]{H. Vincent Poor}(Life Fellow, IEEE) received the Ph.D. degree in EECS from Princeton University in 1977. From 1977 until 1990, he was on the faculty of the University of Illinois at Urbana-Champaign. Since 1990 he has been on the faculty at Princeton, where he is currently the Michael Henry Strater University Professor. During 2006 to 2016, he served as the dean of Princeton’s School of Engineering and Applied Science. He has also held visiting appointments at several other universities, including most recently at Berkeley and Cambridge. His research interests are in the areas of information theory, machine learning and network science, and their applications in wireless networks, energy systems and related fields. Among his publications in these areas is the recent book \textit{Machine Learning and Wireless Communications}. (Cambridge University Press, 2022). Dr. Poor is a member of the National Academy of Engineering and the National Academy of Sciences and is a foreign member of the Chinese Academy of Sciences, the Royal Society, and other national and international academies. He received the IEEE Alexander Graham Bell Medal in 2017.
\end{IEEEbiography}

\end{document}